\documentclass[journal]{IEEEtran}
\IEEEoverridecommandlockouts

\makeatletter
\def\markboth#1#2{\def\leftmark{\@IEEEcompsoconly{\sffamily}\MakeUppercase{\protect#1}}%
\def\rightmark{\@IEEEcompsoconly{\sffamily}\MakeUppercase{\protect#2}}}
\makeatother

 \PassOptionsToPackage{bookmarks={false}}{hyperref}

\usepackage[utf8x]{inputenc}
\usepackage[english]{babel}
\usepackage{ucs}
\usepackage{amsmath}
\usepackage{amsfonts}
\usepackage{amssymb}
\usepackage{amsthm}
\usepackage{array}
\usepackage{verbatim}
\usepackage{listings}
\usepackage{psfrag}
\usepackage{stfloats}

\usepackage{algorithm}
\usepackage{algorithmic}
\usepackage{url}
  \usepackage{enumerate}
  \usepackage{cite}
\usepackage[usenames,dvipsnames,svgnames,table]{xcolor}
\usepackage{setspace}

\ifCLASSINFOpdf
  \usepackage[pdftex]{graphicx}
  \usepackage{subfigure}
  \usepackage{epstopdf}
  \graphicspath{{./}}
   \DeclareGraphicsExtensions{.pdf}
\else
  \usepackage[dvipdf]{graphicx}
  \usepackage{subfigure}
  \graphicspath{{./figures/}}
   \DeclareGraphicsExtensions{.eps,.ps}
\fi

\newcommand{\Sb}{\mathbf{\Sigma}}

\newcommand{\F}{\mathbf{F}}

\newcommand{\A}{\mathbf{A}}
\newcommand{\B}{\mathbf{B}}
\newcommand{\Cb}{\mathbf{C}}
\newcommand{\D}{\mathbf{D}}

\newcommand{\I}{\mathbf{I}}

\newcommand{\Pb}{\mathbf{P}}
\newcommand{\Phib}{\mathbf{\boldsymbol{\Phi}}}
\newcommand{\Upsb}{\mathbf{\boldsymbol{\Upsilon}}}

\newcommand{\cc}{\mathbf{c}}
\newcommand{\h}{\mathbf{h}}
\newcommand{\x}{\mathbf{x}}

\newcommand{\vv}{\mathbf{v}}

\newcommand{\y}{\mathbf{y}}
\newcommand{\bb}{\mathbf{b}}

\newcommand{\vphi}{\mathbf{\boldsymbol{\phi}}}

\newcommand{\rr}{\mathbf{r}}
\newcommand{\z}{\mathbf{z}}
\newcommand{\ab}{\mathbf{a}}
\newcommand{\pp}{\mathbf{p}}

\newcommand{\tr}{\textnormal{tr}}

\newcommand{\Ex}[2]{{\textnormal{E}_{#1}\left[#2\right]}}

\usepackage{pifont}
\newcommand\Tau{\mathcal{T}}

\newtheorem{definition}{Definition}
\newtheorem{lemma}{Lemma}

\theoremstyle{plain}

   \definecolor{blueH3}{rgb}{0,.5,1}
   \definecolor{blueH2}{rgb}{0,0.25,0.75}
   \definecolor{blueH1}{rgb}{0,0,0.5}   
   \definecolor{grayOldText}{rgb}{.5,.5,.5}
   \definecolor{VCobalt}{HTML}{005682}
   \definecolor{TZTeal}{HTML}{008080}
   \definecolor{KYJade}{HTML}{008151}
   \definecolor{ARust}{HTML}{a10000}
   \definecolor{FFucsia}{HTML}{7000c3}
   
\newcommand{\CASE}[1]{\STATE \textbf{case} #1\textbf{:} \begin{ALC@g}}
\newcommand{\ENDCASE}{\end{ALC@g}}

\newcommand{\DEFAULT}{\STATE \textbf{default:} \begin{ALC@g}}
\newcommand{\ENDDEFAULT}{\end{ALC@g}}
\newcommand{\DEFAULTLINE}[1]{\STATE \textbf{default:} }

\newcounter{MYtempeqncnt}

\begin{document}
\sloppy

\title{Compressed Sensing Channel Estimation for OFDM with non-Gaussian Multipath Gains}

  \author{
  Felipe G\'omez-Cuba, ~\IEEEmembership{Member,~IEEE}; Andrea J. Goldsmith,~\IEEEmembership{Fellow,~IEEE}
  \thanks{ 
    Part of this work appeared in IEEE ICC 2019 \cite{gomezcuba2019CSicc}.  
    Felipe Gomez-Cuba is with Dipartimento Di Ingegneria Dell'Informazione, University of Padova, Italy,  Email: \texttt{gomezcuba@dei.unipd.it}. Andrea Goldsmith is with the Electrical Engineering department, Stanford University, USA,  Email: \texttt{andreag@stanford.edu}. This project has received funding from the European Union's Horizon 2020 research and innovation programme under the Marie Sk\l{}odowska-Curie grant agreement No 704837 and from NSF grants ECCS-BSF1609695 and CCF-1320628.}
}
\maketitle

\begin{abstract}
This paper analyzes the impact of non-Gaussian multipath component (MPC) amplitude distributions on the performance of Compressed Sensing (CS) channel estimators for OFDM systems. The number of \textit{dominant MPCs} that any CS algorithm needs to estimate in order to accurately represent the channel is characterized. This number relates to a Compressibility Index (CI) of the channel that depends on the fourth moment of the MPC amplitude distribution. A connection between the Mean Squared Error (MSE) of any CS estimation algorithm and the MPC amplitude distribution fourth moment is revealed that shows a smaller number of MPCs is needed to well-estimate channels when these components have large fourth moment amplitude gains. The analytical results are validated via simulations for channels with lognormal MPCs such as the NYU mmWave channel model. These simulations show that when the MPC amplitude distribution has a high fourth moment, the well known CS algorithm of Orthogonal Matching Pursuit performs almost identically to the Basis Pursuit De-Noising algorithm with a much lower computational cost.
\end{abstract}


\section{Introduction}
As demand for higher data rates continues to grow, the wireless industry is looking to support it through increased bandwidth transmissions. 
Extensive measurements indicate that many wireless channels, especially those in the mmWave band \cite{Mathew2016,Specification2017}, experience \textit{sparse scattering} in which only a few multipath reflections reach the receiver with arrival delays spread over a relatively long time interval. Orthogonal Frequency Division Multiplexing (OFDM) systems separate this frequency-selective channel into a series of frequency-flat subcarriers with scalar gains \cite{Negi1998,Coleri2002,Morelli2001}. When the number of multipath components (MPC) is significantly smaller than the number of channel taps in the time domain or subcarriers in the frequency domain, the OFDM channel does not exhibit independent coefficients.

Wireless channel models that do not assume sparsity typically assume a complex Gaussian distribution for each channel sample. This is justified for digital systems with moderate sample rates where several MPCs arrive during the same sampling period, and the Central Limit Theorem (CLT) is invoked over their sum \cite{goldsmith2005book}. Channel estimation in OFDM systems without sparsity has been extensively studied, for example in \cite{Negi1998,Coleri2002,Morelli2001}.

Compressive Sensing (CS) is a general framework for the estimation of sparse vectors from linear measurements \cite{Duarte2011}. CS has been extensively applied to sparse channel estimation, under different notions of sparsity in the channel model. In particular, some CS OFDM channel estimation studies consider a channel to be sparse when many of its discrete taps are zero and the other taps are the sum of many MPCs of similar delay. In this model the non-zero channel taps follow a Gaussian distribution by the CLT, and hence their amplitude is Rayleigh distributed \cite{Simeone2004,Karabulut2004,ChunJungWu2005,He2010,Haupt2010,Qi2011a,Berger2010,Meng2012,Taubock2008,Qi2011,Berger2010a,Prasad2014,Wang2016,Ling2016}. In some massive MIMO frequency-flat channel models, sparsity is defined with respect to the angles of departure and arrival \cite{Rao2014,Kokshoorn2015,Alkhateeb2015,Rodriguez2016,Schniter,Gao2016freq,Marzi2016a,Lee2016,Liu2016b}. In frequency-selective mmWave systems with sparsity in the number of MPCs (i.e. joint angle and delay sparsity), CS estimators have been designed in the time domain \cite{7953407,Mo2017}, OFDM \cite{rodriguez2017frequency} and joint time-frequency domains \cite{Venugopal2017}. In general the CS frequency-selective mmWave channel estimation literature has either assumed the Gaussian MPC gain model with Rayleigh MPC amplitudes \cite{Mo2017,rodriguez2017frequency}, or has not specified the MPC amplitude distribution \cite{7953407,Venugopal2017}. In general, neither the frequency-flat \cite{Kokshoorn2015,Alkhateeb2015,Rodriguez2016,Schniter,Gao2016freq,Marzi2016a,Lee2016} nor the frequency-selective CS channel estimation literature \cite{Simeone2004,Karabulut2004,ChunJungWu2005,He2010,Haupt2010,Qi2011a,Berger2010,Meng2012,Taubock2008,Qi2011,Berger2010a,Prasad2014,Wang2016,Ling2016,7953407,Mo2017,rodriguez2017frequency,Venugopal2017} have discussed the impact that different MPC amplitude distributions might have on estimator performance.

In some wireless channel models and measurements it has been shown  that the gain distribution for MPCs differs from the Gaussian assumed in \cite{Simeone2004,Karabulut2004,ChunJungWu2005,He2010,Haupt2010,Qi2011a,Berger2010,Meng2012,Taubock2008,Qi2011,Berger2010a,Prasad2014,Wang2016,Ling2016,Rao2014,Kokshoorn2015,Alkhateeb2015,Rodriguez2016,Schniter,Gao2016freq,Marzi2016a,Lee2016,Liu2016b,7953407,Mo2017,rodriguez2017frequency,Venugopal2017}. For instance in mmWave channels the sampling rate is higher than sub-6GHz systems, and individual MPCs arrive at different sampling intervals. A lognormal model for amplitudes, rather than Rayleigh, is prescribed in the mmWave  NYU Wireless \cite{Mathew2016} and 3GPP New Radio \cite{Specification2017} channel models. The physical motivation to assume a lognormal MPC amplitude distribution in other wireless channels is provided in \cite{Coulson1998,bristolmeasurement}.

The main contribution of our paper is analyzing the impact of non-Gaussian MPC gain distributions on different CS estimators. We show that the number of dominant MPCs needed for CS estimators to accurately estimate the channel decreases for MPC amplitude distributions with a high fourth moment, such as lognormal MPC amplitudes. We also show that the OMP algorithm performs well in channels with this characteristic. We illustrate our analysis with numerical results for the NYU mmWave MPC model \cite{Mathew2016}, which has lognormal MPC amplitudes.

Classic CS results are often non-Bayesian in the sense of assuming a fixed estimated-vector structure. For example in \cite[Theorems 6,7]{Duarte2011} the error is bounded if the vector is strictly sparse with a specific number of non-zeros. Nevertheless, the relation between CS estimator performance and random vector distributions is not yet fully understood. Recently \cite{Gribonval2012} established that the ``compressibility'' of a random i.i.d. vector distribution under Gaussian observation matrices depends on the second and fourth moments. Our analysis extends the observations of  \cite{Gribonval2012} to the case of sparse wireless channel CS estimation and studies how this problem is influenced by the distribution of the MPC amplitudes. Instead of considering the estimation of an i.i.d. vector with a Gaussian sensing matrix as in \cite{Gribonval2012}, our observation matrix is the Discrete Fourier Transform (DFT) of a sparse frequency-selective channel with arbitrary MPC delays. In CS the term \textit{superresolution} describes sensing matrices that enable vector reconstruction surpassing the Shannon-Nyquist sampling limits \cite{Duarte2011}. A low Mutual Incoherence (MI) and the Restricted Isometry Property (RIP) are defined in \cite{Duarte2011} as the sufficient conditions for classic CS error analyses. Unfortunately the DFT matrix with delay superresolution does not have these properties \cite{Berger2010a}. This is a difference with both Gaussian matrices \cite{Gribonval2012} 
and with the pilot matrices in frequency-flat massive MIMO CS channel estimation \cite{Kokshoorn2015,Alkhateeb2015,Rodriguez2016,Schniter,Gao2016freq,Marzi2016a,Lee2016,Liu2016b}. The first part of our analysis is a characterization of the MSE of the OMP algorithm without relying on low MI or the RIP properties. The second part of our analysis establishes the connection between the performance of \textbf{any} CS algorithm and a ``compressibility index'' that is related to the fourth-moment. 

We study the OMP algorithm in depth for our general channel model and show it outperforms a maximum-likelihood (ML) non-sparse estimator. Other CS algorithms such as Basis Pursuit De-Noising (BPDN) have shown good performance for channel estimation in some sparse channels, however OMP has been shown to outperform BPDN if the MPC amplitudes are a geometric decaying sequence \cite{Schnass2018}. Based on this we show that for high fourth moment MPC amplitude distributions OMP performs at least as well as BPDN, with much lower computational cost. Although we compare the performance of OMP and BPDN for a given channel model, the main goal of our paper is not to design the best CS channel estimator for a given channel model, but rather to provide a compressibility analysis applicable to \textbf{any} CS channel estimation algorithm under different MPC amplitude models. The interpretation
of compressibility in relation to the stop condition of greedy algorithms, including OMP, was earlier pointed out
in \cite{Gribonval2012}.

Our analysis indicates that the error of an OMP estimator grows linearly with the number of iterations. We define a decaying ``residual'' function $\rho(d)$ that measures the channel power that is not accounted for when OMP retrieves $d$ MPCs. The faster $\rho(d)$ decays, the fewer iterations OMP performs, and the lower the error. We characterize the ``compressibility'' of an arbitrary channel vector through 
an ``oracle estimator'' benchmark proposed in \cite{Gribonval2012}, with residual $\overline\rho(d)\leq\rho(d)$. This oracle lower bound holds for any CS algorithms and is not limited to OMP. Since the channel is not an i.i.d. vector, we study the decay-speed of $\overline\rho(d)$ substituting the fourth moment metric in \cite{Gribonval2012} by the Compressibility Index (CI) \cite{Jain1984}. The CI metric is a measure of the equality of elements in a set because it is inverse to the empirical fourth moment. Our results show that channel models with a high fourth moment in the MPC amplitude are much more ``compressible'' in the sense of having a lower CI (faster decay in $\overline{\rho}(d)$) than the Gaussian distributed MPC gain models in \cite{Simeone2004,Karabulut2004,ChunJungWu2005,He2010,Haupt2010,Qi2011a,Berger2010,Meng2012,Taubock2008,Qi2011,Berger2010a,Prasad2014,Wang2016,Ling2016,Rao2014,Kokshoorn2015,Alkhateeb2015,Rodriguez2016,Schniter,Gao2016freq,Marzi2016a,Lee2016,Liu2016b,Mo2017,7953407,rodriguez2017frequency,Venugopal2017}. We illustrate our results with numerical examples for the lognormal-amplitude mmWave channels models in \cite{Mathew2016,Specification2017}. In particular we show that in a Non-Line-of-Sight (NLOS) outdoor dense urban cellular mmWave link at $60$ m distance based on the NYU mmWave MPC model \cite{Mathew2016}, the MSE of OMP OFDM channel estimation is $3$ dB better than if the mmWave MPCs have Rayleigh distributed amplitudes as in \cite{rodriguez2017frequency,Mo2017}. We also show that BPDN is only $.86$ dB better than OMP in terms of MSE with a much higher computational complexity.

The remainder of this paper is structured as follows. Sec. \ref{sec:model} describes the system model. Sec. \ref{sec:LSML} defines two Maximum Likelihood (ML) estimator benchmarks. Sec. \ref{sec:CS} describes the OMP channel estimator. Sec. \ref{sec:OMPstop} contains the random vector compressibility analysis. Sec. \ref{sec:num} provides numerical results.
Sec. \ref{sec:receiver} shows that the estimation differences affect the performance of OFDM MMSE receivers.
Sec. \ref{sec:conclusions} concludes the paper.

\subsection{Notation}
Calligraphic letters denote sets. $|\mathcal{A}|$ denotes the cardinality of set $\mathcal{A}$. Bold uppercase letters denote matrices. $\|\A\|_n=\left(\sum_{i,j}|a_{i,j}|^n\right)^{\frac{1}{n}}$ is the $\ell_n$ Entrywise norm of $\A$. $\A=\B\cdot\Cb$ is the element-wise product. $\A^H$ is the Hermitian and $\A^{\dag}$ the Moore-Penrose pseudo-inverse $(\A^H\A)^{-1}\A^{H}$.
Bold lowercase letters denote vectors with $\|\ab\|_n$, $\ab=\bb\cdot\cc$, $\ab^{H}$, and $\ab^{\dag}$ as in a single-column matrix. When not specified, $n=2$ so $\|\ab\|=\|\ab\|_2=\sqrt{\ab^H\ab}$ is the vector length and $\|\A\|=\sqrt{\tr{\A^H\A}}$ the Frobenius norm. $\sim$ is the ``distributed as'' sign. $\mathcal{CN}(\mu,\sigma^2)$ is the Gaussian and $U(a,b)$ the uniform distribution.  We use the approximate inequality relation $a\gtrsim b$ to denote that ``$a$ is either greater or approximately equal to $b$'', that is, either $a>b$ or ($a<b$ and $a\simeq b$).

\section{System Model}
\label{sec:model}
\subsection{Multipath Time-Domain Channel Model}

We consider a time-invariant discrete-time equivalent channel (DEC) with Finite Impulse Response (FIR) length $M$ as in \cite{rodriguez2017frequency,7953407,Mo2017,Venugopal2017}. The channel is the sum of $L$ planar waves with fixed amplitudes $\{\alpha_{\ell}\}_{\ell=1}^{L}$, phases $\{\phi_{\ell}\}_{\ell=1}^{L}$ and delays $\{\tau_{\ell}\}_{\ell=1}^{L}$. The discrete-time conversion is modeled by a transmit pulse $p(t)$ and the sampling period $T$.
\begin{equation}
 \label{eq:chanM}
 h_M[n]=\sum_{\ell=1}^{L}\alpha_{\ell}e^{j\phi_{\ell}}p(nT-\tau_{\ell}),\quad n\in[0,M-1]
\end{equation}

The set of delays $\{\tau_{\ell}\}_{\ell=1}^{L}$ is ordered ($\tau_\ell>\tau_{\ell-1}$) and aligned to zero ($\tau_1=0$). The maximum delay spread is $D_s=\max \tau_{\ell}=\tau_L$. Typical choices of $p(t)$ have a peak at $t=0$ and weak infinite tails, so $M= \lceil D_s/T\rceil$ guarantees that all the MPCs are contained in the \textcolor{black}{FIR} DEC. 

We can rewrite \eqref{eq:chanM} as a vector. To do so, we first define a size-$M$ time-domain channel vector
$$\h_M\triangleq\left(h_M[0],h_M[1],h_M[2],\dots,h_M[M-1]\right)^T.$$
Second, we define the size-$M$ pulse-delay vector
$$\pp(\tau)\triangleq\left(p(-\tau),p(T-\tau),\dots,p((M-1)T-\tau)\right)^T$$
where we verified by simulation that for certain pulses $p(t)$ such as a sinc or Raised Cosine, $\|\pp(\tau)\|^2\simeq1$ and, if $\tau_\ell-\tau_{\ell-1}>T/2$, then $\pp(\tau_\ell)^H\pp(\tau_{\ell-1})\ll1$.

Third, we define the size-$M\times L$ pulse-delay matrix $\Pb_{\{\tau_\ell\}_{\ell=1}^{L}}\triangleq\left(\pp(\tau_1),\pp(\tau_2),\dots,\pp(\tau_L)\right)$
where $\|\Pb_{\{\tau_\ell\}_{\ell=1}^{L}}\|^2\simeq L$. If $L<M$ and $\tau_\ell\neq\tau_{\ell'}\;\forall \ell\neq\ell'$ then $\Pb_{\{\tau_\ell\}_{\ell=1}^{L}}$ is full column rank.

Finally, we define the size-$L$ MPC complex gain vector
$$\ab\triangleq\left(\alpha_1e^{j\phi_1},\alpha_2e^{j\phi_2},\alpha_3e^{j\phi_3},\dots,\alpha_Le^{j\phi_L}\right)^T.$$
Using these definitions, the sum in \eqref{eq:chanM} can be written as the following matrix-vector expression:
\begin{equation}
 \label{eq:chanPb}
 \h_M=\sum_{\ell=1}^{L}\pp(\tau_\ell)\alpha_\ell e^{j\phi_\ell}=\Pb_{\{\tau_\ell\}_{\ell=1}^{L}}\ab.
\end{equation}

We assume $L\leq M$. The number of MPCs $L$ and the sets $\{\alpha_{\ell}\}_{\ell=1}^{L}$, $\{\phi_{\ell}\}_{\ell=1}^{L}$ and $\{\tau_{\ell}\}_{\ell=1}^{L}$ are generated following explicit random distributions described below, and we apply \eqref{eq:chanPb} to obtain the channel. Due to this, $\h_M$ has a probability density function that is too cumbersome to write explicitly. If $p(t)$ has a peak at $t=0$ and weak tails, the larger the fourth moment in the distribution of $\{\alpha_{\ell}\}_{\ell=1}^{L}$, the more unevenly the energy is distributed among the coefficients of $\h_M$.

We assume that the carrier frequency is $f_c>10/T$. Compared to the wavelength $\lambda$ the delays satisfy $\tau_\ell-\tau_{\ell-1}\gg 1/f_c=\lambda/c$. Due to this the phases $\{\phi_{\ell}\}_{\ell=1}^{L}$ are independent $\sim U(0,2\pi)$.

The MPC delays $\{\tau_{\ell}\}_{\ell=1}^{L}$ follow a Poisson Arrival Process (PAP). 
In a Uniform PAP the inter-arrival gaps $\Upsilon_i=\tau_{i}-\tau_{i-1}$ follow a memoryless distribution. However, measurements shown that the MPC PAP in wireless systems has memory \cite{Turin1972}. Hence $\{\tau_{\ell}\}_{\ell=1}^{L}$ may also follow the more general Non-Uniform PAP (NUPAP) with MPCs grouped in ``time clusters'' \cite{Saleh1987}. 

Our main result studies differences in estimation MSE for different arbitrary distributions of $\{\alpha_{\ell}\}_{\ell=1}^{L}$. Most CS references assume some variant of $\alpha_{\ell}e^{j\phi_{\ell}}\sim\mathcal{CN}(0,1/L)$ \cite{Simeone2004,Karabulut2004,ChunJungWu2005,He2010,Haupt2010,Qi2011a,Berger2010,Meng2012,Taubock2008,Qi2011,Berger2010a,Prasad2014,Wang2016,Ling2016,Rao2014,Kokshoorn2015,Alkhateeb2015,Rodriguez2016,Schniter,Gao2016freq,Marzi2016a,Lee2016,Liu2016b,Mo2017,7953407,rodriguez2017frequency,Venugopal2017}, hence $\alpha_{\ell}$ would be Rayleigh. In \cite{Mathew2016,Specification2017}, instead, $\{\alpha_{\ell}\}_{\ell=1}^{L}$ follow a normalized lognormal distribution with a delay-dependent mean, where the unnormalized amplitudes $\overline{\alpha}_{\ell}$ satisfy $\log \overline{\alpha}_{\ell} = -\tau_\ell/\Gamma +\zeta_\ell$. Here $\Gamma$ is the mean received power decay with delay \cite{Saleh1987} and $\zeta_\ell\sim  \mathcal{N}(0,\sigma_\alpha^2)$ is a shadowing distribution that randomizes the amplitude \cite{Coulson1998}. 
The normalization $\alpha_\ell=\sqrt{\frac{P_{recv}}{\sum_{\ell=1}^{L} \overline{\alpha}_\ell^2}}\overline{\alpha}_\ell$ where $P_{recv}$ is the total received power is applied for consistency with macroscopic shadowing and pathloss. Since the delay NUPAP features clustering, this makes both the delays \textit{and the amplitudes} dependent across different MPCs. This is modeled in \cite{Mathew2016,Specification2017} with two normalized lognormals: one to divide $P_{recv}$ among the clusters and one to divide each cluster power among its MPCs.

A key difference between non-Gaussian and Gaussian sparse channel models is that in the latter typically $T\gg\tau_{\ell}-\tau_{\ell-1}$ when the $\ell$-th and $(\ell-1)$-th MPC belong to the same cluster. As several MPCs arrive during the same sampling intervals, the CLT is invoked to model the non-zero coefficients of $\h_M$ as Gaussian, avoiding the need to explicitly model the set $\{\alpha_{\ell}\}_{\ell=1}^{L}$.

\subsection{OFDM Channel Model and Pilot Scheme}

\begin{figure}
 \centering
 \includegraphics[width=.95\columnwidth]{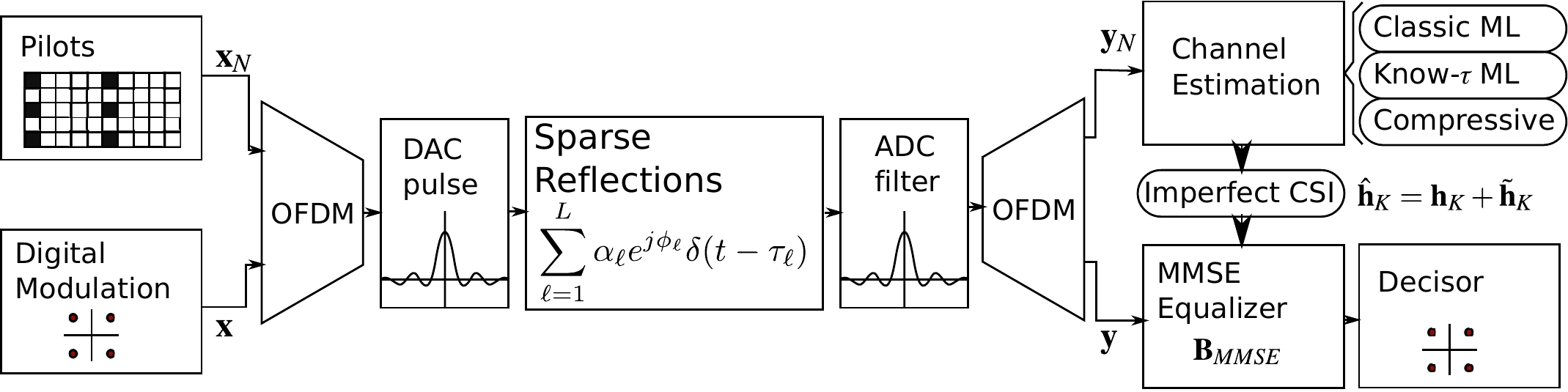}
 \caption{OFDM system with pilot subcarriers, a linear receiver with imperfect CSI, and a sparse multipath channel.}
 \label{fig:sysmod}
\end{figure}

We assume the OFDM system in Fig. \ref{fig:sysmod}. The number of subcarriers (DFT size) is $K\geq M$, the Cyclic Prefix (CP) length is $M$, and the OFDM signal is sent through the \textcolor{black}{FIR} DEC (1). At the receiver, the frequency-domain received signal is the $K$-coefficient vector
\begin{equation}
 \label{eq:chanOFDM}
 \y=\D(\x)\h_K+\z,
\end{equation}
where $\z$ is Additive White Gaussian Noise (AWGN) with variance $\sigma^2\I_K$, $\x$ are the $K$ IDFT inputs, $\D(\x)$ is the diagonal matrix containing the vector $\x$ in the main diagonal, and $\h_K$ is the size-$K$ DFT of \eqref{eq:chanPb}.
Ignoring the last $K-M$ columns of the DFT matrix, we can write this as
\begin{equation}
 \label{eq:chanK}
 \h_K=\F_{K,M} \h_M
\end{equation}
where the rectangular matrix $\F_{K,M}$ contains the first $M$ columns of the size-$K$ normalized DFT matrix that satisfies $\F_{K,M}^H\F_{K,M}=\I_M$ but $\F_{K,M}\F_{K,M}^H$ is not an identity matrix.

In practical scenarios the phases $\{\phi_{\ell}\}_{\ell=1}^{L}$ change faster than $\{\frac{\tau_{\ell}}{T}\}_{\ell=1}^{L}$ and $\{\alpha_{\ell}\}_{\ell=1}^{L}$. $\h_K$ is assumed time-invariant when $\{\phi_{\ell}\}_{\ell=1}^{L}$ does not change significantly during a block of several OFDM frames. In the first OFDM frame of each block, we assume a ``comb'' pilot pattern \cite{Negi1998} with $N$ pilot subcarriers is used, where $K\geq N\geq M$.
We denote the size-$N$ vector of transmitted pilots by $\x_N$ (a subset of the coefficients of  $\x$ where pilots are transmitted). Assuming $K/N\in\mathbb{N}$, we define the vector of received coefficients in pilot subcarriers during the first OFDM frame as
$$\y_N=\D(\x_N)\h_{N/K}+\z_N.$$
Denoting by $\F_{N/K,M}$ the submatrix that contains the \textit{first} $M$ columns and the \textit{alternated} $N$\textit{-out-of-}$K$ rows of the DFT we define the frequency-domain channel vector at the pilots as
\begin{equation}
 \label{eq:chanN}
 \h_{N/K}\triangleq\F_{N/K,M}\h_M
\end{equation}
where $\F_{N/K,M}^H\F_{N/K,M}=\frac{N}{K}\I_M$. The full channel $\h_K$ can be recovered by inverting \eqref{eq:chanN} and substituting in \eqref{eq:chanK} producing $\h_K=\F_{K,M}\h_M=\F_{K,M}(\F_{N/K,M})^{\dag} \h_{N/K}=\F_{K,M}\F_{N/K,M}^H\frac{K}{N}\h_{N/K}$.

Substituting \eqref{eq:chanPb} into \eqref{eq:chanK} we can write two alternative linear representations of $\h_K$ as follows.
\begin{equation}
 \label{eq:chanMPClinear}
 \underset{\textnormal{frequency}}{\underbrace{\h_K}}=\F_{K,M}\underset{\textnormal{discrete-time}}{\underbrace{\h_M}}=\F_{K,M}\Pb_{\{\tau_\ell\}_{\ell=1}^{L}}\underset{\textnormal{MPCs}}{\underbrace{\ab}},
\end{equation}
Using these two linear representations, non-sparse and sparse channel estimators can be written, respectively, as we will describe in more detail in the next section.

\section{LS/ML Estimation Benchmarks}
\label{sec:LSML}

\subsection{Conventional Discrete-time-domain LS/ML estimation}
\label{sec:LSMLconventional}
A non-sparse ML estimator of $\h_K$ is given in \cite{Morelli2001}, using $\h_K=\F_{K,M}\h_M$ but not requiring that $\h_M$ is sparse. This non-Bayesian estimator is the best we can do in non-sparse channels when the distribution of	 $\h_M$ is unknown or untractable, and we adopt it as the ``non-sparse benchmark''.

For $M\leq N \leq K$, the ML estimator of $\h_K$ subject to a linear constraint $\h_K=\F_{K,M}\h_M$ is a Least Squares (LS) estimator of $\h_M$ multiplied by $\F_{K,M}$ to reconstruct $\h_K$ \cite{Morelli2001}:

\begin{equation}
\label{eq:MLM}
\hat{\h}_M^{\textnormal{ML-}M}=(\D(\x_N)\F_{N/K,M})^\dag\y_N,
\end{equation}
\begin{equation}
\label{eq:MLMK}
\hat{\h}_K^{\textnormal{ML-}M}=\F_{K,M}\hat{\h}_M^{\textnormal{ML-}M},
\end{equation}
where the error can be expressed as
$$\tilde{\h}_K^{\textnormal{ML-}M}\triangleq\hat{\h}_K^{\textnormal{ML-}M}-\h_K=\F_{K,M}(\D(\x_N)\F_{N/K,M})^\dag\z_N.$$
Since $\z_N$ is AWGN the error is Gaussian, zero mean (unbiased) and the MSE is the variance
\begin{equation}
\label{eq:errMLMK}
\begin{split}
\nu_{\textnormal{ML-}M}^2&=\frac{\Ex{\z}{\|\F_{K,M}(\D(\x_N)\F_{N/K,M})^\dag\z_N\|^2}}{K}\\
&= \frac{\sigma^2 M}{N^2}\tr\{(\D(\x_N)^H\D(\x_N))^{-1}\}\\
&\geq\frac{M}{N}\sigma^2,
\end{split}
\end{equation}
where we solve ${\displaystyle \min_{\x_N:\|\x_N\|^2\leq N}} \tr\{(\D(\x_N)^H\D(\x_N))^{-1}\}$ to get the minimum MSE pilots, which have unit-amplitude coefficients $|x[k]|=1\;\forall k$ \cite{Negi1998}. Our simulations confirm this result (Fig. \ref{fig:mse128}). 

It is important to remark that although other non-sparse OFDM estimation methods are possible, this ML scheme is the ``best non-sparse estimator'' in our case. In \cite{Negi1998} a MMSE channel estimator is derived, but it requires the assumption that $\h_M$ is Gaussian distributed. And in \cite{Coleri2002} direct LS estimation of $\h_N$ is performed, $\hat{\h}_N=\D(\x_N)^{-1}\y_N$, followed by a frequency-domain interpolation of $\h_K$ from $\h_N$. Although this LS estimator reduces complexity, it is not the ML estimator and it has $\frac{N}{M}$ times more error variance than our benchmark.

Some receivers need to know the covariance matrix of the error. Since $\z_N$ is Gaussian and $M\leq N\leq K$ the error vector $\tilde{\h}_K^{\textnormal{ML-}K}$ is Gaussian and has a rank-M covariance matrix 
\begin{equation}
\label{eq:covMLMK}
\Sb_{\tilde{\h}_K}=\nu_{\textnormal{ML-}M}^2\F_{K,M}\F_{K,M}^H.
\end{equation}

We observe that the coefficients of this matrix can be written as
\begin{equation}
\label{eq:covSinc}
\Sigma_{\tilde{\h}_K}^{i,j}=\nu_{\textnormal{ML-}M}^2\frac{\sin(2\pi \frac{M(i-j)}{K})}{\sin(2\pi \frac{(i-j)}{K})}e^{j2\pi \frac{(M-1)(i-j)}{K})}.
\end{equation}
This expression stems from the time-domain interpolation \eqref{eq:MLMK}, which in frequency domain can be represented by the expression $\F_{K,M}\F_{M,M}^H$ that behaves as a periodical-sinc interpolation filter.

\subsection{Sparse MPC-domain LS/ML estimation}
\label{sec:LSMLMPC}

We define our second benchmark assuming the channel is sparse and $\{\tau_\ell\}_{\ell=1}^{L}$ is perfectly known to the receiver (``genie-aided''). 
With minor changes to Sec. \ref{sec:LSMLconventional} we can derive a LS estimator of the vector $\ab$ and use $\F_{K,M}\Pb_{\{\tau_\ell\}_{\ell=1}^{L}}$ to recover the frequency-domain channel as
\begin{equation}
\label{eq:MLa}
\hat{\ab}^{\textnormal{ML-}{\{\tau_\ell\}_{\ell=1}^{L}}}= (\D(\x_N)\F_{N/K,M}\Pb_{\{\tau_\ell\}_{\ell=1}^{L}})^{\dag}\y_N
\end{equation}
\begin{equation}
\label{eq:MLaK}
\hat{\h}_K^{\textnormal{ML-}{\{\tau_\ell\}_{\ell=1}^{L}}}=\F_{K,M}\Pb_{\{\tau_\ell\}_{\ell=1}^{L}}\hat{\ab}^{\textnormal{ML-}{\{\tau_\ell\}_{\ell=1}^{L}}}
\end{equation}

For AWGN $\z_N$ the error is Gaussian, zero mean (unbiased) and the MSE is the variance
\begin{equation}
\begin{split}
\label{eq:errMLaK}
\nu_{\textnormal{ML-}{\{\tau_\ell\}_{\ell=1}^{L}}}^2
&\geq \frac{L}{N}\sigma^2,\\
\end{split}
\end{equation}
where again the equality is achieved for $|x[k]|=1\;\forall k$. Noting that $L<M$ the genie-aided benchmark offers a gain versus the non-sparse benchmark of $M/L$. For example choosing $M=N$ for minimal pilot overhead and the values $L=30$ and $M=128$ results in a gain of $6$ dB. We verified this result in numerical simulations (Fig. \ref{fig:mse128}).

For the error convariance matrix we get
\begin{equation}
\label{eq:covMLaK}
\Sb_{\tilde{\h}_K}=\frac{\sigma^2}{N}\F_{K,M}\Pb_{\{\tau_\ell\}_{\ell=1}^{L}}\Pb_{\{\tau_\ell\}_{\ell=1}^{L}}^\dag\F_{K,M}^H,
\end{equation}
where we observe the DFT $\F_{K,M}$ on both sides, which as we argued can be interpreted as a M-to-K periodical sinc interpolation filter \eqref{eq:covSinc}. However, in-between $\F_{K,M}$ and $\F_{K,M}^H$ we have the $M\times M$ rank-$L$ matrix $\Pb_{\{\tau_\ell\}_{\ell=1}^{L}}\Pb_{\{\tau_\ell\}_{\ell=1}^{L}}^\dag$ instead of $\I_M$. This ``inner'' term makes the error covariance matrix rank $L$ instead of $M$. And the zeros of \eqref{eq:covSinc} do not necessarily correspond with zeros in \eqref{eq:covMLaK}. Furthermore, the covariance matrix depends on $\{\tau_\ell\}_{\ell=1}^L$, so the error covariance matrix varies for different realizations of the sparse channel.

\section{Hybrid CS/LS estimator of Sparse MPCs}
\label{sec:CS}

In this section we assume that the pilot sequence $\x_N$ has unit-amplitude symbols as in the previous section. Without loss of generality we describe the estimator for $\D(\x_N)=\I_N$ where the results remain valid if we first compute $\y_N'=\D(\x_N)^H\y_N$ for other pilot sequences.

In practice the delays $\{\tau_\ell\}_{\ell=1}^{L}$ are not known to the receiver and it cannot use $\Pb_{\{\tau_\ell\}_{\ell=1}^{L}}$ to implement the sparse LS estimator \eqref{eq:MLa}, \eqref{eq:MLaK}. To design a practical estimator, we assume the channel is sparse in the sense that a linear relation $\h_K=\F_{K,M}\Pb_{\{\tau_\ell\}_{\ell=1}^{L}}\ab$ exists but its matrix is unknown, and we define a hybrid two-step estimator as follows:
\begin{enumerate}
 \item We use any CS method to estimate the delays of the channel, denoted as $\hat{\Tau}\simeq \{\tau_\ell\}_{\ell=1}^{L}$.
 \item We use MPC-domain LS estimation \eqref{eq:MLa} based on the matrix $\Pb_{\hat{\Tau}}$ instead of $\Pb_{\{\tau_\ell\}_{\ell=1}^{L}}$.
\end{enumerate}

We define two entities of interest for this procedure, which remain valid for any CS method.
\begin{definition}
 We define the \textbf{residual} associated with an approximation of the delay set $\hat{\Tau}$ as
 \begin{equation}
\begin{split}
\rho(\hat{\Tau})&\triangleq \frac{\|((\I_K)-\F_{K,M}\Pb_{\hat{\Tau}}\Pb_{\hat{\Tau}}^{\dag}\F_{K,M}^H)\h_K\|^2}{K}\\
&=\frac{\|(\I_M-\Pb_{\hat{\Tau}}\Pb_{\hat{\Tau}}^\dag)\h_M\|^2}{K}.
\end{split}
\end{equation}
\end{definition}

\begin{definition}
\label{def:subsp}
We define the \textbf{subspace} associated with an approximation of the delay set $\hat{\Tau}$ as
  \begin{equation}
\begin{split}
&\mathcal{S}(\F_{K,M}\Pb_{\hat{\Tau}})\triangleq\\
&\quad\left\{\vv\in \mathbb{C}^K: \|(\I_K-\F_{K,M}\Pb_{\hat{\Tau}}\Pb_{\hat{\Tau}}^{\dag}\F_{K,M}^H)\vv\|^2=0\right\}.
\end{split}
\end{equation}
\end{definition}\

For any algorithm adopted in the first step, the error can be expressed as the projection of the noise over the subspace $\mathcal{S}(\F_{K,M}\Pb_{\hat{\Tau}})$ plus $K\rho(\hat{\Tau})$. To show this, we define the projection of $\h_K$ on $\mathcal{S}(\F_{K,M}\Pb_{\hat{\Tau}})$ as $\h_S=\F_{K,M}\Pb_{\hat{\Tau}}(\F_{K,M}\Pb_{\hat{\Tau}})^\dag\h_K$. Since $\h_S\in\mathcal{S}(\F_{K,M}\Pb_{\hat{\Tau}})$, there exists some vector $\bb\in\mathbb{C}^{|\hat{\Tau}|}$ such that $\h_S=\textcolor{black}{\F_{K,M}}\Pb_{\hat{\Tau}}\bb$. Step 2 uses \eqref{eq:MLa} to estimate $\bb$ and $\h_S$, producing $\hat{\bb}=(\D(\x_N)\F_{N/K,M}\Pb_{\hat{\Tau}})^{\dag}\y_N$ and $\hat{\h}_S=\textcolor{black}{\F_{K,M}}\Pb_{\hat{\Tau}}\hat{\bb}$. Since $\h_S$ is a projection of $\h_K$, we adopt $\hat{\h}_S$ as estimator of $\h_K$. We can write the error as $\tilde{\h}_K=\h_K-\hat{\h}_S=(\h_K-\h_S)+(\h_S-\hat{\h}_S)=\h_{E}+\tilde{\h}_S.$ Here, $\h_{E}$ is the error of the first step satisfying $\|\h_{E}\|^2=K\rho(\hat\Tau)$ and $\tilde{\h}_S=\h_S-\hat{\h}_S=\F_{K,M}\Pb_{\hat{\Tau}}(\bb-\hat{\bb})$ is the error of the second step. Definition \ref{def:subsp} means that $\h_E$ is orthogonal to $\mathcal{S}(\F_{K,M}\Pb_{\hat{\Tau}})$, and particularly to $\tilde{\h}_S$, and therefore $\|\tilde{\h}_K\|^2=K\rho(\hat{\Tau})+\|\F_{K,M}\Pb_{\hat{\Tau}}(\bb-\hat{\bb})\|^2$.

\subsection{Estimation of $\{\tau_\ell\}_{\ell=1}^{L}$ as a CS Problem}

In the paragraphs above we described the first step abstractly as ``obtain some estimate of the delays $\hat{\Tau}$''. In this section we discuss how this task can be framed as a CS problem. The \textit{delay dictionary set} with size $N_T\geq M$ is defined as $\Tau_{N_T}\triangleq\{n\frac{D_s}{N_T}:n\in[0,N_T-1]\}$. We wish to find a subset $\hat{\Tau}\subset\Tau_{N_T}$ with $\hat{L}=|\hat{\Tau|}$ such that 1) the subspaces $\mathcal{S}(\F_{K,M}\Pb_{\{\tau_\ell\}_{\ell=1}^{L}})$ and $\mathcal{S}(\F_{K,M}\Pb_{\hat{\Tau}})$ are similar in the sense that $\rho(\hat{\Tau})$ is small, and 2) the matrix $\Pb_{\hat{\Tau}}$ is 
a submatrix formed by $\hat{L}$ of the columns of the \textit{delay dictionary matrix} $\Pb_{\Tau_{N_T}}$, with $\hat{L}\ll M$, so that the pseudoinverse of $\Pb_{\hat{\Tau}}$ exists and $\|\tilde{\h}_S\|^2\propto \frac{\hat{L}}{N}$ is small. Ideally, for this we would solve
\begin{equation}
\label{eq:besterror}
\min_{\hat{\Tau}\subset\Tau_{N_T}}\min_{\bb}\Ex{\z}{\|\h_E\|^2+\|\tilde{\h}_S\|^2},
\end{equation}
where increasing $|\hat{\Tau}|$ increases $\|\tilde{\h}_S\|^2$ and decreases $\|\h_E\|^2$. However, $\|\h_E\|^2$ cannot be evaluated without knowing $\h_M$. Instead we consider an approximate error minimization. First we define
\begin{equation}
\begin{split}
\h_{N/K}&=\F_{N/K,M}\Pb_{\hat{\Tau}}\bb\\
&=\F_{N/K,M}\Pb_{\Tau_{N_T}}\bb_{N_T}\\
&=\Phib_{N_T}\bb_{N_T}
\end{split}
\end{equation}
where $\bb$ is a non-sparse size-$\hat{L}$ vector and $\bb_{N_T}$ is the sparse size-$N_T$ vector with the coefficients of $\bb$ in the appropriate places and zeros elsewhere. If we define the matrix $\Phib_{N_T}=\F_{N/K,M}\Pb_{\Tau_{N_T}}$, then identifying the non-zero coefficients of the sparse vector $\bb_{N_T}$ is a classic CS problem.

As an approximation to \eqref{eq:besterror}, we can use the $\ell_0$ minimization
\begin{equation}
\label{eq:bestsparse}
 \min \|\hat{\bb}_{N_T}\|_0 \textnormal{ s.t. } \|\y_N-\Phib_{N_T}\hat{\bb}_{N_T}\|_2^2\leq \xi,
\end{equation}
so that $|\hat{\Tau}|$ is minimized and $\rho(\hat{\Tau})$ is controlled by $\xi$. 
Intuitively, relaxing $\xi$ allows more reduction of $|\hat{\Tau}|$ but increases $\rho(\hat{\Tau})$. The value $\xi=N\sigma^2$ allows $\|\y_N-\Phib_{N_T}\hat{\bb}_{N_T}\|_2^2$ to be of similar power to a size-$N$ noise vector, and brings $\rho(\hat{\Tau})$ below some notion of a noise floor.

By definition when $N_T=M$ the delay dictionary contains the sampling instants $\Tau_{M}=\{0,T,2T\dots (M-1)T\}$ and is a \textit{complete} dictionary without superresolution. On the other hand when $N_T>M$, $\Tau_{N_T}$ is an \textit{overcomplete} dictionary with superresolution. Unfortunately in the general case the matrix $\Phib_{N_T}$ would not meet sufficient conditions to guarantee that the representation of $\h_{N/K}$ as $\Phib_{N_T}\bb_{N_T}$ is unique \cite{Duarte2011}. Thus in CS OFDM channel estimation we cannot have simultaneously superresolution and invoke \cite[Theorems 6,7]{Duarte2011} to guarantee that the indices of the large coefficients of $\bb_{N_T}$ are recovered correctly. Nonetheless, in channel estimation we do not care if $\hat{\Tau}$ contains false-positive errors as long as $\|\tilde{\h}_K\|^2$ is minimized.

By definition $\tau_\ell$ can take any real non-negative value, and therefore any finite dictionary ($N_T<\infty$) would incur some delay discretization error. It is common to disregard this error in CS channel estimator designs \cite{7953407,rodriguez2017frequency,Venugopal2017}, but in this paper we consider the extension to the continuous case, enabling an effectively infinite dictionary $N_T=\infty$. The continuous dictionary introduces no discretization errors and, as far as the dictionary design is concerned, it would be possible to recover $\ab$ and $\Pb_{\{\tau_\ell\}_{1}^{L}}$ exactly. Rather than the finite matrix product $\Phib_{N_T}\bb_{N_T}$, the continuous case considers the integral $\int_{0}^{\infty}\boldsymbol{\phi}_{\infty}(\tau)b_{\infty}(\tau)d\tau=\Pb_{\{\tau_\ell\}_{\ell=1}^{L}}\ab$ where $\boldsymbol{\phi}_{\infty}(\tau)=\F_{N/K,M}\pp(\tau)$ and $b_{\infty}(\tau)=\sum_{\ell=1}^{L}\alpha_\ell e^{j\phi_\ell}\delta(\tau-\tau_\ell)$ are the equivalents of the discrete dictionary $\Phib_{N_T}$ and the sparse vector $\bb_{N_T}$, respectively, and estimation is performed over the function-space with bases $\{\boldsymbol{\phi}_{\infty}(\tau):\tau\in[0,D_s]\}$ and where the integral is the distance.

\subsection{OMP Heuristic to solve \eqref{eq:bestsparse}}

The $\ell_0$ problem \eqref{eq:bestsparse} has combinatorial complexity in the number of columns of $\Phib_{N_T}$. There are many CS algorithms that can be roughly divided into two branches to address this: substitution by an $\ell_1$ problem (BPDN, LASSO, Dantzig-Selector), and greedy approximations of the combinatorial problem (OMP, CoSaMP) \eqref{eq:bestsparse} \cite{Duarte2011}. 
Results in \cite{Schnass2018} have established that BPDN requires a weaker sufficient condition than OMP to recover the sparse support correctly when the non-zero coefficients have similar magnitude. On the other hand OMP requires a weaker condition than BDPN when the non-zero coefficient magnitudes display a geometric decay. Wireless channels with a high fourth moment MPC gain distribution typically generate a fast-decaying collection of MPC amplitudes similar to this second scenario. CoSaMP is faster than OMP because it selects $L$ elements per iteration, but this requires knowing $L$ beforehand, which is random and unknown in our model. LASSO has a similar limitation as it needs to adopt an upper bound of $\|\h_{N/K}\|_1$ as one of the optimization constraints. For these reasons we select OMP for solving the CS approximation problem \eqref{eq:bestsparse}.

In addition, the error analysis of OMP sheds light on the theoretical connection between CS and the ``compressibility'' analysis we give in Section \ref{sec:OMPstop}. OMP is a greedy algorithm that, on each iteration, adds one column to the matrix minimizing the LS projection of $\y_N$. 
The interpretation of compressibility in relation to the stop condition of greedy algorithms was earlier pointed out in \cite{Gribonval2012}. Interpreting OMP as a direct heuristic for \eqref{eq:besterror} instead of \eqref{eq:bestsparse}, each greedy iteration follows the steepest decrease of $\|\h_E\|^2$ and increases $\|\tilde{\h}_S\|^2$, and OMP stops ($\xi$) when the next decrement of $\|\h_E\|^2$ is not worth its associated increase of $\|\tilde{\h}_S\|^2$. 
Intuitively, the choice $\xi=N\sigma^2$ means that the residual $\|\y_{N/K}-\Phib_{N_T}\hat{\bb}_{N_T}\|^2<\xi$ contains only noise and under-the-noise-floor channel coefficients, and additional steps do not lead to a further reduction of the MSE.

 \begin{algorithm}[!t]
\caption{Orthogonal Matching Pursuit with Binary-search Refinement (OMPBR)}
\label{alg:OMP}
\begin{algorithmic}[1]
  \STATE Def. dictionary ${ \Tau_{N_T}=\{n\frac{D_s}{N_T}\}\;\forall n\in\{0\dots N_T-1\}}$
  \STATE Generate $N/K$-FFTs  $\vphi_n=\F_{N/K,M}\pp(\hat{\tau}_n)$
  \STATE Initialize estimate of sparse support set $\hat{\Tau}_0=\emptyset$
  \STATE Initialize residual with data observation $\rr_0=\y_{N/K}$
  \WHILE{$\|\rr_i\|^2>\xi$ and $i<$ max num. iterations}    
    \STATE $\overline{\tau}_i=\{\arg\max |\vphi_n^H\rr_{i-1}|\;\forall \Tau_{N_T}\setminus\mathcal{T}_{i-1}\}$
    \STATE $\mu^*={\displaystyle \arg\max_{\mu \in[\frac{-1}{2},\frac{1}{2}]}} |\pp(\overline{\tau}_i+\mu \frac{D_s}{N_T})^H\F_{N/K,M}^H\rr_{i-1}|$
    \STATE $\hat{\tau}_i=\overline{\tau}_i+\mu^* \frac{D_s}{N_T}$
    \STATE Update estimation of support $\hat{\Tau}_{i}=\hat{\Tau}_{i-1}\cup \{\hat{\tau_i}\}$
    \STATE Update support matrix $\hat{\Phib}_i=\F_{N/K,M}\Pb_{\hat{\Tau}_{i}}$
    \STATE Update LS channel estimator  $\hat{\h}_i=\hat{\Phib}_i\hat{\Phib}_i^\dag\y_{N/K}$
    \STATE Update residual for next step $\rr_{i}=\y_{N/K}-\hat{\h}_i$, 
  \ENDWHILE
\end{algorithmic}
\end{algorithm}

Our main contribution consists in analyzing CS estimators for non-Gaussian MPC gain distributions. To analyze this result when there is no delay discretization error we extend OMP to OMPBR which has a continuous delay dictionary in Algorithm \ref{alg:OMP}, while setting $\mu^*=0$ in line 8 converts the algorithm to ``classic'' OMP.
Lines 6-8 of Alg. \ref{alg:OMP} solve the subproblem
 \begin{equation}
 \label{eq:subprbolembisection}
  \max_{\tau} |\pp(\tau)^H\F_{N/K,M}^H\rr_{i-1}|,
 \end{equation}
 which is non-concave in the interval $[0,D_s]$. Since typical pulses $p(t)$ such as the Raised Cosine are symmetric at $t=0$ and concave in the interval $-T/2,T/2$, we make the assumption that this problem is locally concave and symmetric around a local maximum in small regions we call \textit{delay bins}. First a finite dictionary is used to identify the best bin as in OMP, centered at $\overline{\tau}_i$ with width $\pm\frac{1}{2}\frac{D_s}{N_T}$. Then the delay is refined as $\hat{\tau}_i=\overline{\tau}_i+\mu^* \frac{D_s}{N_T}$ where $\mu^*$ comes from an assumed locally concave and symmetric maximization. We use the Binary-search Local Maximum Refinement described in Alg. \ref{alg:binary}, 
rather than a gradient as \cite{Marzi2016a}, to guarantee that OMPBR is robust in the sense that the result is contained in the bin and never worse than the decision that OMP would make. Lemma \ref{lem:refinement} shows Alg. \ref{alg:binary} is optimal if the target function satisfies the assumptions. 

\begin{algorithm}[!h]
\caption{Binary-search Local Maximum Refinement}
\label{alg:binary}
\begin{algorithmic}[1]
  \STATE Assume a function $f(\mu)$, and initial interval $[\mu_{\min},\mu_{\max}]$, and desired relative resolution $\delta_\mu$
  \STATE Initialize $\mu_0=\mu_{\min}$,  $\mu_1=\mu_{\max}$
  \WHILE{$|\mu_1-\mu_0|>$ $2\delta_\mu$}
    \IF{$f(\mu_0)<f(\mu_1)$}
      \STATE $\mu_0=\frac{\mu_0+\mu_1}{2}$
    \ELSE
      \STATE $\mu_1=\frac{\mu_0+\mu_1}{2}$
    \ENDIF    
  \ENDWHILE
  \STATE $\hat{\mu}=\frac{\mu_0+\mu_1}{2}$
\end{algorithmic}
\end{algorithm}

\begin{lemma}
\label{lem:refinement}
 For any function $f(\mu)$ and interval $[\mu_{\min},\mu_{\max}]$ such that 
 \begin{enumerate}
  \item $f(\mu)$ has a single local maximum in the interval, 
located at $\mu^*\triangleq{\displaystyle \arg\max_{\mu\in[\mu_{\min},\mu_{\max}]}}f(\mu)$,
  \item $f(\mu)$ is strictly increasing for $\mu\leq\mu^*$ 
and decreasing for $\mu\geq\mu^*$
  \item $f(\mu)$ is symmetric with regard to $\mu^*$, that is $f(\mu^*+\Delta\mu)=f(\mu^*-\Delta\mu)\;\forall \Delta\mu\in[0,\frac{\mu_{\max}-\mu_{\min}}{2}]$
 \end{enumerate}
 Algorithm \ref{alg:binary} finds a solution $\hat{\mu}$ such that $|\hat{\mu}-\mu^*|<\delta_\mu$.
\end{lemma}
\begin{proof}
 Due to the symmetry of the strictly increasing function  $f(\mu)$, if $f(\mu_0)<f(\mu_1)$ then $|\mu^*-\mu_0|>|\mu^*-\mu_1|$. Therefore $\mu^*$ is contained in the interval $[\frac{\mu_0+\mu_1}{2},\mu_1]$. The converse holds if $f(\mu_0)>f(\mu_1)$. The algorithm halves the interval $[\mu_0,\mu_1]$ in each iteration, and the interval is guaranteed to contain $\mu^*$. The stop condition guarantees $|\hat{\mu}-\mu^*|<|\mu_0-\mu_1|/2<\delta_\mu$.
\end{proof}

For OMP, line 6 of Alg \ref{alg:OMP} performs $O(N_TM)$ products and the delay dictionary has resolution $D_s/N_T$, whereas for OMPBR, lines 6-8 perform $O((N_T+2\log(1/\delta_\mu))M)$ products where $\delta_\mu<1$ adjusts the desired resolution $\delta_\mu D_s/N_T$. Thus, if we desire to improve the delay resolution, the computational cost increase is linear for OMP and logarithmic for OMPBR.
 
In summary we distinguish three cases of delay dictionary sizes we consider:
\begin{itemize}
 \item $N_T=M$ and $\mu^*=0$: in this case the dictionary is finite and without superresolution as in \cite{Taubock2008,Berger2010a,Qi2011a,7953407}. For this scenario if the columns of $\Phib_M$ are orthogonal, for example if $p(t)$ is a Nyquist pulse, then the greedy OMP solves the $\ell_0$ problem \eqref{eq:bestsparse} \textbf{exactly}.
 \item Finite $N_T> M$ and $\mu^*=0$: this case produces a finite \textit{overcomplete} dictionary with superresolution as in \cite{Berger2010a,Venugopal2017}. In this scenario, OMP is an heuristic of \eqref{eq:bestsparse}, so its error is lower bounded by the $\ell_0$ optimum for the same $N_T>M$. Even if we cannot invoke \cite[Theorems 6,7]{Duarte2011},  when $N_T/M$ is an integer, $\Tau_M\subset \Tau_{N_T}$. And since OMP is a greedy algorithm, using a dictionary size $N_T$, such that $N_T/M\in\mathbb{N}$, achieves lower or equal MSE to that of OMP with size $M$, which is also the $\ell_0$ optimum with $N_T=M$. However, increasing $N_T$ results in a linear growth of computational complexity.
 \item Finite $N_T\geq M$ with $\mu^*$ as in Alg. \ref{alg:OMP}: this case is OMPBR. For very small $\delta_\mu$, OMPBR has an \textit{effectively} continuous dictionary with \textit{effectively} infinite size, while complexity only grows with $\log(\delta_\mu^{-1})$. In addition, in each iteration of OMPBR the greedy decision is never worse than in OMP using the same $N_T$. 
\end{itemize}

\subsection{OMP Error Analysis}
\label{sec:error}
We assume OMP stops after $\hat{L}$ iterations and returns $\Pb_{\hat{\Tau}}\in \mathbb{C}^{M\times\hat{L}}$. Both $\Pb_{\hat{\Tau}_{\hat{L}}}$ and $\hat{L}$ are random variables that depend on $\z_N$. 
The MPC phase distribution is $U(0,2\pi$), and even if the error for a fixed $\h_M$ is not unbiased, the average of the error over the distribution of the channel is zero-mean. For compactness, we denote the residual $\rho(\hat{L})=\rho(\hat{\Tau}_{\hat{L}})$. The MSE is
\begin{equation}
\label{eq:omperror}
\begin{split}
\nu_{OMP}^2
&=\frac{\Ex{\z}{\|\tilde{\h}_S\|^2+\|\h_E\|^2}}{K}
\\&
=\frac{\Ex{\z}{\|\Pb_{\hat{\Tau}_{\hat{L}}}(\hat{\bb}-\bb)\|^2}}{K}+\Ex{\z}{\rho(\hat{L})}\\
\end{split}
\end{equation}
which follows from $\F_{K,M}^H\F_{K,M}=\I_M$ and the definitions of $\rho(\hat{L})$, $\tilde{\h}_S$ and $\h_E$.

The following results are proven in Appendices  \ref{app:firsterm} and  \ref{app:secondterm} and characterize the MSE of OMP:
\begin{enumerate}
 \item 
 In OMP, the first MSE term approaches $\frac{\Ex{}{\hat{L}}}{N}\sigma^2$ when $\sigma^2$ is small, 
 $$\lim_{\sigma^2\to0}\frac{\Ex{\z}{\|\Pb_{\hat{\Tau}_{\hat{L}}}(\hat{\bb}-\bb)\|^2}}{K}-\frac{\Ex{}{\hat{L}}}{N}\sigma^2=0.$$
\item
 When $p(t)$ is a Nyquist pulse and $N_T=M$, in OMP the first MSE term satisfies
 $$\frac{\Ex{\z}{\|\Pb_{\hat{\Tau}_{\hat{L}}}(\hat{\bb}-\bb)\|^2}}{K}\geq\frac{\Ex{}{\hat{L}}}{N}\sigma^2.$$
\item
 If we choose $\xi=N\sigma^2$ the first and the second error terms are approximately equal.
 $$\frac{\Ex{\z}{\|\tilde{\h}_S\|^2}}{K}=\frac{\Ex{\z}{\|\Pb_{\hat{\Tau}_{\hat{L}}}(\hat{\bb}-\bb)\|^2}}{K}\simeq\Ex{\z}{\rho(\hat{L})}$$
\end{enumerate}

Putting everything together we get that the MSE can be approximated as
\begin{equation}
 \label{eq:errOMP}
 \begin{split}
 \nu_{OMP}^2&\simeq\frac{2\Ex{\z}{\|\tilde{\h}_S\|^2}}{N}\\
 &\simeq\frac{2\Ex{\z}{\rho(\hat{L})}}{N}\\
 &\stackrel{\sigma^2\ll1}{\xrightarrow{\hspace{.7cm}}}\frac{2\Ex{\z}{\hat{L}}}{N}\sigma^2.
 \end{split}
\end{equation}
where the asymptotic result is a tight lower bound if $M=N_T$ and $p(t)$ is a Nyquist pulse.

Our simulations verify that $\hat{L}$ grows with SNR and that $\nu_{OMP}^2$ converges to $\frac{2\Ex{\z}{\hat{L}}}{N}\sigma^2$ as the SNR increases (Figs. \ref{fig:ntaps128}, \ref{fig:mse128}). We also conjecture that this is a lower bound in the general case. Two arguments support this conjecture: First, 
if we let $\sigma^2\to\infty$ then $\y_N\simeq\z_N$ and OMP becomes a greedy algorithm that selects the noise dimensions in decreasing order of power, picking up above-average noise power. Second, the simulation results displayed such behavior.

The analysis shows that when OMP stops, the number of recovered MPCs $\hat{L}$ is determined by the meeting point between a term that decays as $\rho(\hat{L})$ and a term that grows linearly with $\hat{L}\sigma^2$.
We can interpret $\rho(d)$ as a random non-increasing function of $d$ with a decay that depends on the inequality between the coefficients of $\h_M$. 
If we characterize $\rho(d)$ in relation to the distribution of $\h_M$, we can use the stop condition result in reverse to deduce $\Ex{\z}{\hat{L}}$ from $\h_M$ and assess how the MSE varies with different MPC amplitude distributions. In the next section we propose a ``compressibility index'' of $\h_M$  related to 
a lower bound of $\rho(d)$ that holds for any CS algorithm. 

\subsection{Covariance Matrix}

The covariance matrix of the estimator is also given by a two-term expression. For compact notation, we define the complementary matrix $\Upsb_{\hat{\Tau}}=(\I_M-\Pb_{\hat{\Tau}}\Pb_{\hat{\Tau}}^\dag)\Pb_{\{\tau_\ell\}_{\ell= 1}^{L}}$, which satisfies $\h_E=\h_K-\h_S=\F_{K,M}\Upsb_{\hat{\Tau}}\ab$; and define the noise projection $\z_b=\hat{\bb}-\bb=\Pb_{\hat{\Tau}}^\dag\frac{K}{N}\F_{N/K,M}^H\z_N$. Using this,
the error covariance matrix of any two-step estimator may be denoted as follows
\begin{equation}
 \label{eq:covOMP}
 \begin{split}
 &\Sb_{\tilde{\h}}=\\
 &\quad\F_{K,M}\Ex{\z}{\Pb_{\hat{\Tau}}(\z_b\z_b^H)\Pb_{\hat{\Tau}}^H+\Upsb_{\hat{\Tau}}(\ab\ab^H)\Upsb_{\hat{\Tau}}^H}\F_{K,M}^H.
 \end{split}
\end{equation}
Again, we observe the DFT matrix $\F_{K,M}$ in both sides, that is interpreted as a 
$M$ to $K$ periodic sinc interpolation filter \eqref{eq:covSinc}. The terms in-between depend on the pulse-delay matrix similarly to \eqref{eq:covMLaK}. The main difference is that in the case of CS estimation the ``inner'' matrix contains two terms and neither is associated with a Gaussian distribution due to the dependency of $\hat{\Tau}$ with $\z_N$. In addition $\Upsb_{\hat{\Tau}}\ab$ cannot be evaluated without knowing $\Pb_{\{\tau_\ell\}_{\ell= 1}^{L}}$ and $\ab$. This means that in practical receivers that need the covariance of the channel estimation error, the use of CS estimators instead of conventional estimators raises the need to design approximations of $\Sb_{\tilde{\h}}$.

\section{Compressibility Analysis}
\label{sec:OMPstop}

Since the coefficients of $\h_M$ in \eqref{eq:chanPb} are not i.i.d. and we do not consider a Gaussian sensing matrix, the results relating compressibility and the fourth-moment in \cite{Gribonval2012} do not apply directly. We define the Compressibility Index (CI) of an arbitrary channel vector $\h_M$ as follows:
\begin{equation}
 \label{eq:CIdef}
 \textnormal{CI}(\h_M)=\frac{(\sum_{n=0}^{M-1}|h_M[n]|^2)^2}{M(\sum_{n=0}^{M-1}|h_M[n]|^4)}.
\end{equation}
Which is based on Jain's fairness index for schedulers \cite{Jain1984} and has the following properties:
\begin{itemize}
 \item If exactly $L\in\{1\dots M\}$ elements of $\h_M$ are non-zero and equal, the CI is $L/M$
 \item 
The inverse of the CI can be rewritten as a normalized estimator of the fourth moment of a sequence. So if $\vv$ is a size-$M$ i.i.d. vector following a zero-mean distribution $f(v[n])$, then 
 \begin{equation}
 \begin{split}
 \lim_{M\to\infty}\frac{1}{\textnormal{CI}(\vv)}&=\lim_{M\to\infty}\frac{\left(\frac{1}{M}\sum_{n=1}^M|v[n]|^4\right)}{\left(\frac{1}{M}\sum_{n=1}^M|v[n]|^2\right)^2}\\
 &=\frac{\Ex{\vv}{|v[n]|^4}}{\Ex{\vv}{|v[n]|^2}^2}\\
 &=\kappa(v[n])
 \end{split}
 \end{equation}
 where $\kappa$ is the \textit{kurtosis} of the distribution $f(v[n])$. Therefore for large i.i.d. vectors, the CI becomes the normalized fourth moment, and is consistent with the metric used in \cite{Gribonval2012}.
 \item For channels modeled as a sum of planar waves as in \eqref{eq:chanM} the CI of $\h_M=\Pb_{\{\tau_\ell\}_{\ell=1}^L}\ab$ is strongly related to the CI of the set $\{\alpha_{\ell}\}_{\ell=1}^{L}$. In particular, if the delays $\{\tau_\ell\}_{\ell=1}^{L}$ are exact multiples of $T$ and $p(t)$ is a Nyquist pulse, then $\textnormal{CI}(\h_M)=\frac{L}{M}\textnormal{CI}(\ab)$. 
 \item The CI is invariant to scale, so for any scalar $\lambda\in\mathbb{C}$ we get $\textnormal{CI}(\lambda\h_M)=\textnormal{CI}(\h_M)$. This means that we can disregard the normalization of the set $\{\alpha_{\ell}\}_{\ell=1}^{L}$ specified by \cite{Mathew2016,Specification2017} and $\textnormal{CI}(\h_M)$ depends on the fourth moment of the lognormal distribution that generates $\{\overline\alpha_{\ell}\}_{\ell=1}^{L}$.
\end{itemize}

The CI is interpreted in the sense that a vector with score $\mathrm{CI}(\h_M)$ has the same score as a size-$M$ vector with exactly $M\mathrm{CI}(\h_M)$ non-zero coefficients of equal power. For example, a Gaussian i.i.d. vector with $\kappa=2$, when $M$ is very large, has the same CI as a strictly-sparse vector with $M/2$ non-zero equal elements; and a Bernoulli-Lognormal i.i.d. vector with $L/M$ non-zero coefficients, when $\kappa=e^{4\sigma^2}$ for $\sigma^2=\log(10)/4$, has a CI equal to a strictly sparse vector containing exactly $L/10$ coefficients of equal magnitude.

First, we consider an ``oracle'' CS benchmark. For a given delay dictionary $\Phib_{N_T}$ the ``oracle'' reveals in $d$ iterations the set of $d$ columns of $\Phib_{N_T}$ that minimizes $\|\h_{N/K}-\Phib_{N_T}\hat{\bb}_{N_T}\|_2^2$. By definition this ``oracle'' has minimal residual among any CS algorithm that estimates $\hat{\Tau}$ using the same dictionary, that is the oracle residual satisfies $\overline\rho(d)\leq\rho(\hat{\Tau})$ for $d=|\hat{\Tau}|$.

Second, we focus on the case where $N_T=M$ and $p(t)$ is a Nyquist pulse to show the connection between $\textnormal{CI}(\h_M)$ and $\rho(d)$. We define the sorted sequence of magnitudes of $\h_M$ as
$$m_i=\begin{cases}
       \max (\{|h_M[n]|^2\}\textcolor{black}{_{n=0}^{M-1}})& i=1\\
       \max (\{|h_M[n]|^2\}\textcolor{black}{_{n=0}^{M-1}}\setminus \bigcup_{j=1}^{i-1}\{m_j\})& i>1\\
      \end{cases}
$$
where $m_1$ is the power of the largest coefficient, $m_2$ the second, and so on. Since $\sum_{\ell=1}^{L}\alpha_\ell^2=1$ and $\|\pp(\tau)\|^2=1$, we get $\|\h_M\|^2=1$ and $1-\sum_{i=1}^{d-1}m_i=\sum_{i=d}^{M}m_i$. In $d$ iterations the oracle recovers $\{m_i\}_{i=1}^{d}$ perfectly, and the residual is the recursive function
\begin{equation}
\label{eq:defrhoid}
\begin{split}
\overline{\rho}(d)&\triangleq \frac{1}{K}\left(1-\sum_{i=1}^{d}m_i\right)\\
&=\overline{\rho}(d-1)\left(1-\frac{m_{d}}{\sum_{i=d}^{M}m_i}\right).
\end{split}
\end{equation}

Finally, we relate $\overline{\rho}(d)$ to the CI.
We first define a set containing the residual channel coefficient powers after the $d$ strongest are perfectly recovered, denoted by
$$\mathcal{R}_d\triangleq \{m_j\}_{j=d+1}^{M}=\{|h_M[n]|^2\}\textcolor{black}{_{n=0}^{M-1}}\setminus \{m_j\}_{j=1}^{d}.$$
The CI of this set evaluates to
$$\textnormal{CI}(\mathcal{R}_d)=\frac{(\sum_{i=d+1}^{M}m_i)^2}{(M-d)\sum_{i=d+1}^{M}m_i^2}.$$

By definition the CI satisfies the following inequalities
\begin{equation}
\label{eq:CIbounds}
\small
\frac{1}{(M-d)\sqrt{\textnormal{CI}(\mathcal{R}_d)}}\leq\frac{m_{d+1}}{\sum_{i=d+1}^{M}m_i}\leq\frac{1}{\sqrt{(M-d)\textnormal{CI}(\mathcal{R}_d)}}.
\end{equation}
Introducing the right hand side of \eqref{eq:CIbounds} into \eqref{eq:defrhoid} we get
\begin{equation}
\label{eq:rhoLB1}
\begin{split}
  \rho(d)\geq\overline{\rho}(d)
  &\geq\prod_{i=0}^{d-1}\left(1-\frac{1}{\sqrt{(M-i)\textnormal{CI}(\mathcal{R}_i)}}\right).
  \end{split}
\end{equation}

So we see that, if the $\textnormal{CI}(\mathcal{R}_i)$'s are high, ${\rho}(d)$ cannot decay fast. The left hand side of \eqref{eq:CIbounds} can be used to write a converse (if $\textnormal{CI}(\mathcal{R}_i)$'s are low $\overline{\rho}(d)$ must decay fast).
We plot the empirical values of \eqref{eq:rhoLB1} in Fig. \ref{fig:boundsrho} averaged over $\h_M$, where we see that it is a lower bound of $\overline{\rho}(d)$.

This bound depends on $\textnormal{CI}(\mathcal{R}_i)$ for $0\leq i\leq d-1$, while we want a bound depending on $\textnormal{CI}(\h_M)=\textnormal{CI}(\mathcal{R}_0)$. We introduce the following assumption: When $d\ll M-d$ we assume that removing the largest element increases the CI by at least $\frac{M-d+1}{M-d}$, i.e.
\begin{equation}
\label{eq:assumptionCI}
\textnormal{CI}(\mathcal{R}_d)\gtrsim\frac{M-d+1}{M-d}\textnormal{CI}(\mathcal{R}_{d-1})
\end{equation}
This assumption was tested numerically and seems to hold in an overwhelming majority of the realizations of the the NYU mmWave MPC model \cite{Mathew2016}, as shown in Fig. \ref{fig:histogramCI}. Recursively replacing $\textnormal{CI}(\mathcal{R}_d)$ with $\textnormal{CI}(\mathcal{R}_{d-1})$ until $CI(\h_M)=\textnormal{CI}(\mathcal{R}_0)$ we can write the following
{\small \begin{equation}
\label{eq:rhoLB2}
\overline{\rho}(d)\gtrsim\left(1-\frac{1}{\sqrt{M \textnormal{CI}(\h_M)}}\right)^d.
\end{equation}}
This approximation shows how $\textnormal{CI}(\h_M)$ influences ``compressibility'': the lower the compressibility index $\textnormal{CI}(\h_M)$, the faster the geometric decay of (an approximation of) $\overline{\rho}(d)$. This suggests that if \eqref{eq:assumptionCI} was a strict equality instead of an approximation we would see OMP outperform BPDN as established in \cite{Schnass2018}. We represented \eqref{eq:rhoLB2} in Fig. \ref{fig:boundsrho}, where we see that it approximates \eqref{eq:rhoLB1} and lower bounds $\rho(d)$. Fig. \ref{fig:L1fig} compares OMP and BPDN, where the MSE difference was less than $1.5$ dB, but OMP does not fully dominate as in \cite{Schnass2018}, since \eqref{eq:rhoLB2} is an approximation.

As we said $CI(\h_M)$ is connected to $\{\alpha_\ell\}_{\ell=1}^{L}$ and if $p(t)$ is Nyquist and the delays are exact multiples of $T$ then $\textnormal{CI}(\h_M)=\frac{L}{M}\textnormal{CI}(\{\alpha_\ell\}_{\ell=1}^{L})$. Since we can ignore normalization \eqref{eq:rhoLB2} becomes
\begin{equation}
\label{eq:rhoLB3}\overline{\rho}(d)\gtrsim\left(1-\frac{1}{\sqrt{L \textnormal{CI}(\{\overline\alpha_\ell\}_{\ell=1}^{L})}}\right)^d.
\end{equation}
We represented the simulation values of \eqref{eq:rhoLB3} in Fig. \ref{fig:boundsrho}. We see that even for the actual NYU mmWave MPC model the approximation $\textnormal{CI}(\h_M)\simeq\frac{L}{M}\textnormal{CI}(\{\overline\alpha_\ell\}_{\ell=1}^{L})$ approximates \eqref{eq:rhoLB2} closely.

Finally, if the marginal distributions of $\overline\alpha_\ell$ and $\tau_\ell$ for different $\ell$ are identical and $L$ is large, 
\begin{equation}\left(1-\frac{1}{\sqrt{L \textnormal{CI}(\{\overline\alpha_\ell\}_{\ell=1}^{L})}}\right)^d\approx\left(1-\sqrt{\frac{\kappa(\overline{\alpha})}{L}}\right)^d,
\end{equation}
where $\kappa(\overline{\alpha})$ is the kurtosis of the marginal distribution $f(\overline{\alpha})=\Ex{\tau_\ell}{f(\overline{\alpha}_\ell|\tau_\ell)}$. Our simulations show that the cumulative density function (CDF) of $\mathrm{CI}(\h_M)$ changes for different distributions of $\{\overline\alpha_\ell\}_{\ell=1}^{L}$, and probability concentrates in lower values of $\mathrm{CI}(\h_M)$ when $\{\overline\alpha_\ell\}_{\ell=1}^{L}$ are lognormal distributed or depend on $\{\tau_\ell\}_{\ell=1}^L$ as in the NYU mmWave MPC model (Fig. \ref{fig:compress}).
  
In conclusion, when the MPC amplitude distribution kurtosis increases, the CI of the channel is smaller, $\overline{\rho}(d)$ decays faster, the ratio $\hat{L}/L$ is lower and $\h_M$ is more ``compressible'' in the sense that CS algorithms can estimate $\h_M$ by retrieving lower number of dominant MPCs. This at the same time reduces the MSE. The ratio $\hat{L}/L$ characterizes a form of \textit{statistical sparsity} due to the fourth moment of the distribution that is distinct from the \textit{physical sparsity} characterized by the ratio $L/M$ and motivated by the fact that there are fewer MPCs than channel taps. We represented $\Ex{\h_M}{\overline{\rho}(d)}$ in Fig. \ref{fig:boundsrho} for the NYU mmWave MPC model \cite{Mathew2016}, Rayleigh MPC amplitudes, and an i.i.d. Gaussian non-sparse channel, where indeed in the NYU mmWave MPC model \cite{Mathew2016} the $\overline\rho(d)$ decays much faster, making the necessary number of iterations $\Ex{\h_M}{\hat{L}}$ much lower.

A limitation of the ``oracle estimator'' analysis is that it lower bounds the residual of practical CS schemes, $\overline{\rho}(d)\leq \rho(d)$. Therefore, there is no use in further developing the left hand side of \eqref{eq:CIbounds}, and in fact steps \eqref{eq:assumptionCI} and \eqref{eq:rhoLB2} do not apply to the upper bound case. Thus, we only proved the claim in one direction, i.e. a large fourth moment in the distribution of $\overline{\alpha}_\ell$ is \textit{necessary} for statistical sparsity (small $\hat{L}/L$), but we have not shown if a high $\kappa(\overline{\alpha})$ is \textit{sufficient} for small $\hat{L}/L$. However, we conjecture that this holds due to the following observations: $\textnormal{CI}(\mathcal{R}_d)$ also can bound $\overline{\rho}(d)$ from above using the left hand side of  \eqref{eq:CIbounds}, the compressibility of generalized-Gaussian distributed i.i.d. vectors improves as the fourth moment increases \cite[Fig. 6]{Gribonval2012}, and we have always observed statistical sparsity for high $\kappa(\overline{\alpha})$ in our simulations.

We note that other metrics of ``channel compressibility'' could have been adopted instead of the fourth-moment based CI. We adopt this metric because the fourth moment has arisen as a  prominent metric to measure compressibility or more general unevenness in statistics in other analyses that are similar to ours. This includes works on compressed sensing, time-varying fading channel statistics, and non-coherent channel capacity \cite{Gribonval2012,Rao2015,fgomezUnified} (see Appendix \ref{app:otherworks}).
 
\section{Numeric Validation}
\label{sec:num}

We simulate an OFDM system with $T=2.5$ ns as in \cite{Mathew2016} and $p(t)=\textnormal{sinc}(t/T)$. This makes the bandwidth $B=400$MHz. We choose the CP length $M=128$, for a realistic maximum delay spread $D_s=TM=320$ ns \cite{Mathew2016}, and use $N=M=128$ pilots. We demonstrate our analysis for $10^3$ realizations of the NYU mmWave MPC model \cite{Mathew2016} for a $60$ m non-line-of-sight link with $f_c=28$ GHz. Simulation results for the 3GPP model \cite{Specification2017} were similar. For a transmitted power in the range $1$ to $10$ W at this distance the median SNR ranges from $-4$ to $6$ dB. Since $\lambda\simeq1$ cm, we choose the DFT size $K=512$, so the frame duration is $KT=1.28$ $\mu$s and the phases vary very little over ten frames for speeds up to $100$ m/s.
\subsection{OMP Error Analysis Validation}

First we look at $\hat{L}$. In Fig \ref{fig:ntaps128} we represent $\Ex{\z,\h_M}{\hat{L}}$ vs SNR for three estimators: no-superresolution OMP ($N_T=M$), conventional superresolution OMP ($N_T=4M$), and our OMP Binary-search Refinement proposal ($N_T=M$ and $\delta_\mu=10^{-2}$, equivalent to OMP with $N_T=100M$). The NYU mmWave MPC model \cite{Mathew2016} generates $L$ as a random number that does not depend on SNR, and in the simulation the average was $\Ex{L}{L}=32$, represented by a red line in Fig \ref{fig:ntaps128}. The bars show the number of MPCs retrieved by OMP. At SNR$=0$ dB no-superresolution OMP estimates about 8 MPCs, whereas both algorithms with superresolution obtain 5. At high SNR orthogonal OMP overshoots to 35 due to its insufficient delay resolution, whereas the dictionaries with superresolution estimate about 20 dominant MPCs. While obtaining a very similar result, OMP was $4\times$ slower with $N_T=4M$ whereas the run time of OMPBR was almost identical to OMP with $N_T=M$. For practical SNR values, a number of MPCs of the NYU mmWave MPC model are too weak to be recovered by OMP, and $\hat{L}<L$.

\begin{figure}
\centering
 \subfigure[Mean value of $\hat{L}$ and std. deviation for OMP vs SNR]{
 \includegraphics[width=.9\columnwidth]{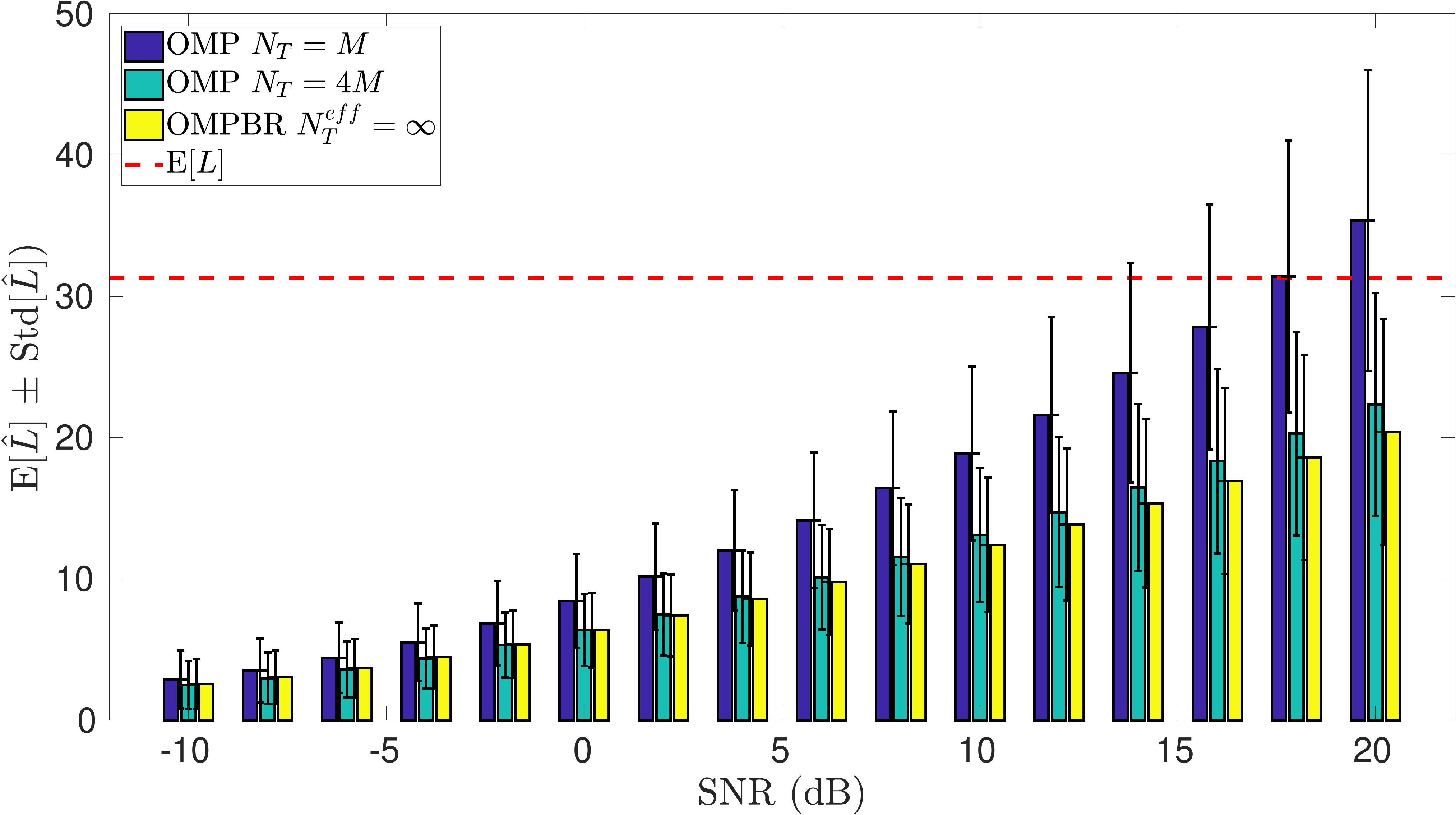}
 \label{fig:ntaps128}
 }
 \hspace{.05in}
 \subfigure[Error variance ($\nu^2$) for different channel estimators]{
 \includegraphics[width=.95\columnwidth]{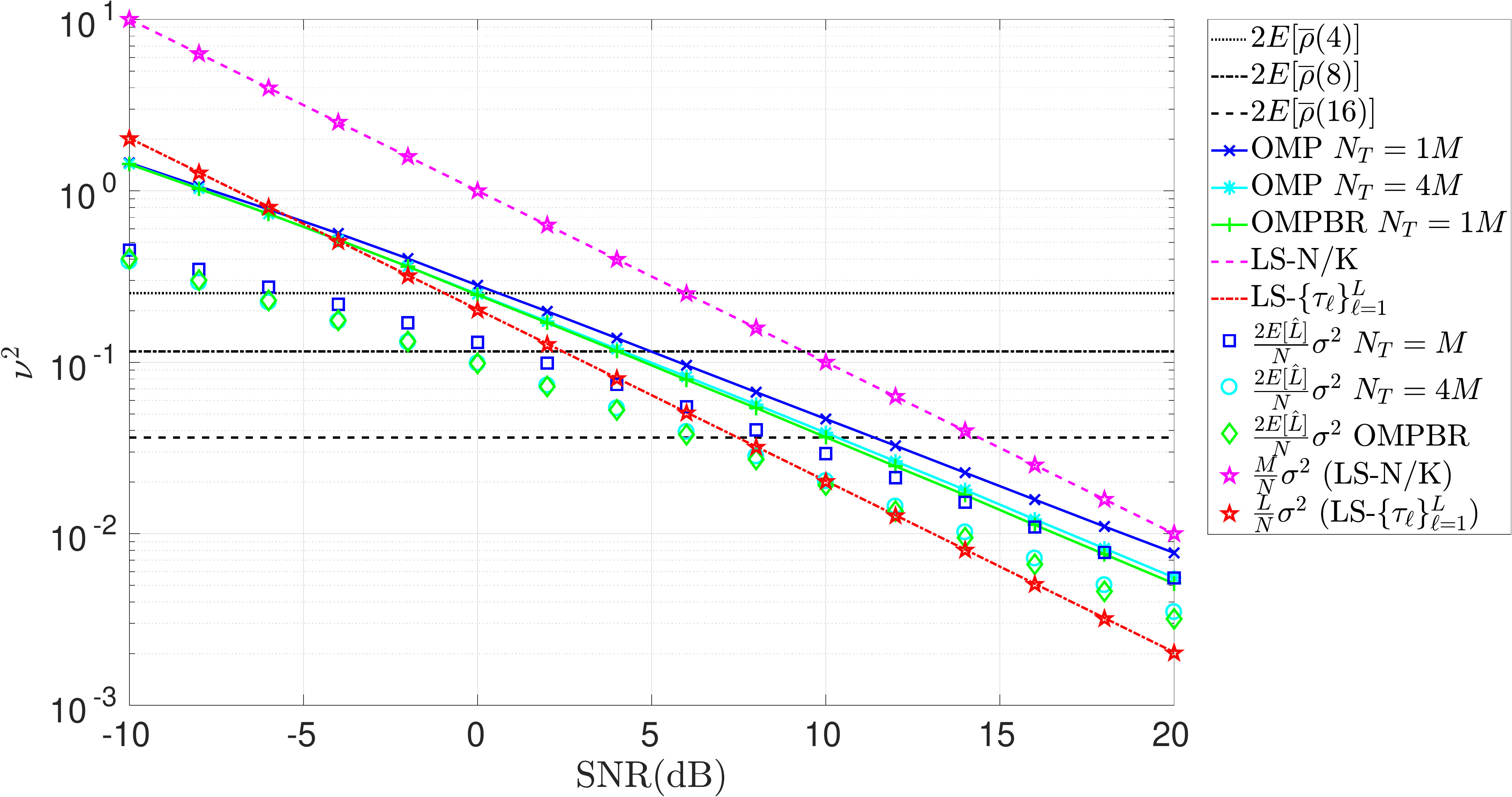}
 \label{fig:mse128}
 }
 \caption{Numerical validation of OMP estimator error analysis.}
 \label{fig:estimperf}
\end{figure}

We represent the MSE ($\nu^2$) in Fig. \ref{fig:mse128}. The CS estimation error averaged over $\h_M$ is unbiased and $\|\h_M\|^2={\sum_{\ell=1}^L}\alpha_\ell^2=1$, so the Normalized MSE (NMSE), the MSE, and the variance are equal. The purple line corresponds to non-sparse LS-ML estimator benchmark defined in Sec. \ref{sec:LSMLconventional}. Its numerical result coincides with our theoretical prediction \eqref{eq:errMLMK}, represented by star-shaped purple bullets in the figure. We represent with a red line the genie-aided LS-ML sparse estimator benchmark defined in Sec. \ref{sec:LSMLMPC}, where we see again the error coincides with our theoretical prediction  \eqref{eq:errMLaK} (red star bullets). Since $\Ex{\h_M}{L}=32$ and $N=M=128$, the sparse benchmark has a gain of approximately $6$ dB. Next, we have the OMP algorithm with $N_T=M$,  $N_T=4M$, and OMPBR. The error variance analysis in Sec. \ref{sec:CS} predicts the error converges to $\frac{\Ex{\z,\h_M}{2\hat{L}}}{N}\sigma^2$ in the high-SNR regime. We represent the empirical error with solid green, cyan and blue lines, and we represent with squared, circular and diamond bullets of the same colors the theoretical lower bound of the error evaluated taking the empirical values of $\Ex{\z,\h_M}{\hat{L}}$ previously displayed in Fig \ref{fig:ntaps128}. We confirm that the theoretical results approximate the empirical error from below and become tighter as the SNR increases. The CS estimators outperform the non-sparse estimator benchmark but fail to match the genie-aided sparse estimator benchmark. At low SNRs OMP displays a lower $\Ex{\z,\h_M}{\hat{L}}$ and achieves a greater advantage over the non-sparse benchmark.

As a reference, we also represented in Fig. \ref{fig:mse128} the lower bound to the residual function $2\overline{\rho}(d)$ as defined in Sec. \ref{sec:OMPstop} for several values of $d$. Our error analysis showed that OMP stops for a value of $\hat{L}$ where approximately $\Ex{\z,\h_M}{\rho(\hat{L})}\simeq\frac{\hat{L}}{N}\sigma^2$. We verify that as the SNR increases the MSE of the OMP estimators decreases slower than $\sigma^2$. As the OMP MSE decreases it crosses the horizontal lines corresponding to $\overline{\rho}(d)$ for increasing values of $d$, consistently with the observation that OMP retrieves more MPCs at high SNR in Fig. \ref{fig:ntaps128}.

\subsection{Stopping point of OMP vs $\textnormal{CI}(\h_M)$}

We display the decaying function $\rho(d)$, related to channel ``compressibility'', and the approximations we derived using the Compressibility Index of $\h_M$ in Fig. \ref{fig:boundsrho}. We depict the empirical average of the function $\Ex{\h_M}{\overline{\rho}(d)}$ versus $d$ for  the NYU mmWave MPC model using a solid red line. We also depict the lower bounds \eqref{eq:rhoLB1} and \eqref{eq:rhoLB2} using purple and blue lines, respectively. The green line represents a Bernoulli-Lognormal simplified-mmWave MPC model, where $\textnormal{CI}(\h_M)$ is exactly $\frac{L}{M}\textnormal{CI}(\{\alpha_\ell\}_{\ell=1}^{L})$.  We verify the $\Ex{\h_M}{\overline{\rho}(d)}$ decays fast for the the NYU mmWave MPC model and is lower bounded by our approximations. Comparatively, the dash-dotted and dashed black lines represent $\Ex{\h_M}{\overline{\rho}(d)}$ for Rayleigh-distributed $\{\alpha_\ell\}_{\ell=1}^{L}$ and for non-sparse Gaussian vectors, showing that for these low-kurtosis distributions $\overline{\rho}(d)$ decays much more slowly. A reference black dotted line representing $d/N$ in Fig. \ref{fig:boundsrho} illustrates the interpretation of $\xi$ in OMP. For each different MPC model, the points where $\rho(d)$ meets this increasing line indicate approximately the number of dominant MPCs, which is much higher for low-kurtosis channels.

\begin{figure}
 \centering 
 \subfigure[Relation between $\Ex{\h_M}{\overline{\rho}(d)}$ for a mmWave channel, a Rayleigh channel, and the bounds obtained in Sec. \ref{sec:OMPstop}.]{
 \includegraphics[width=.95\columnwidth]{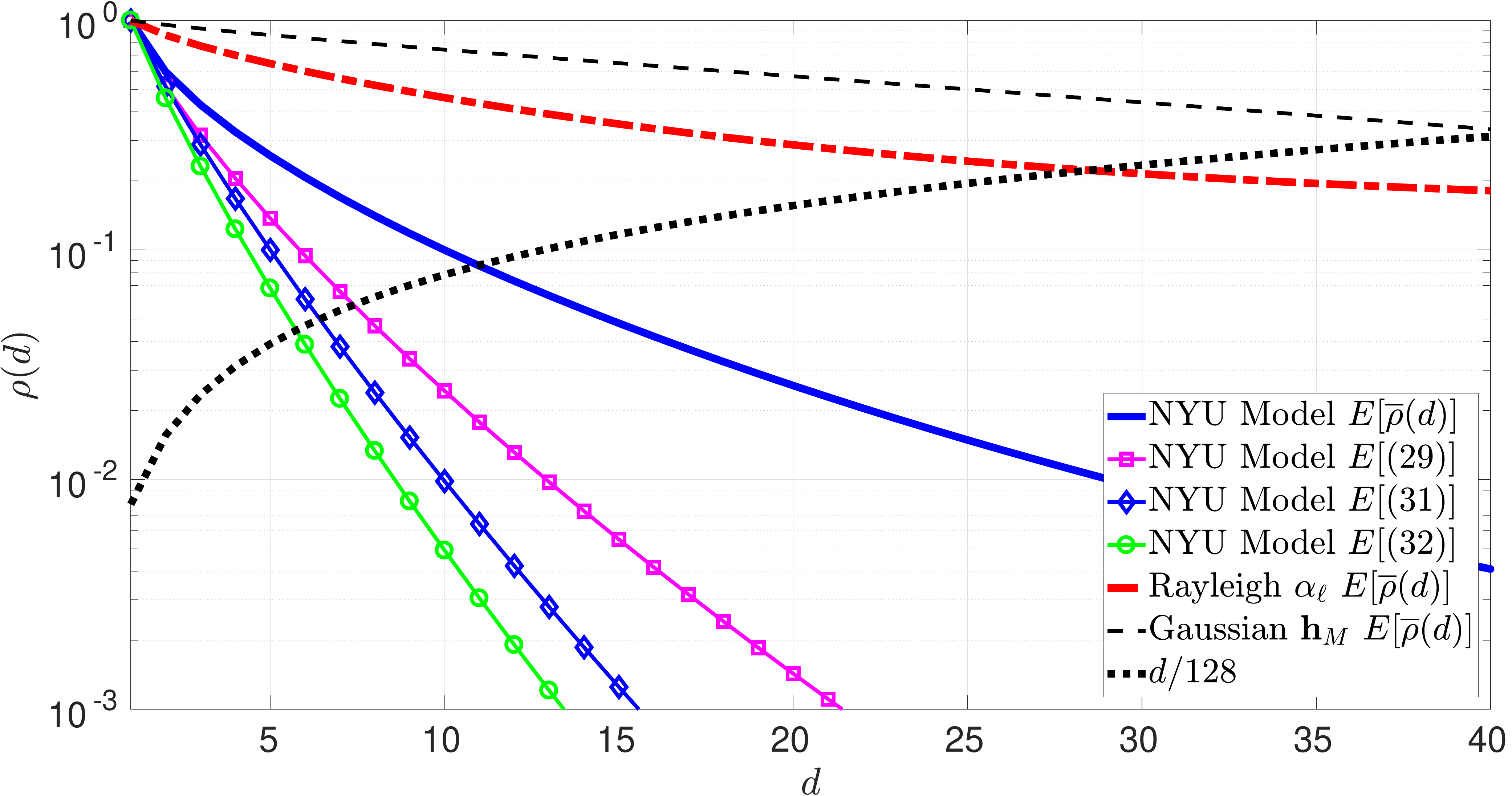}
 \label{fig:boundsrho}
 }
 \hspace{.05in}
 \subfigure[Histogram verification of the assumption $CI(\mathcal{R}_d)\stackrel{L\gg1}{\gtrsim}\frac{M-d+1}{M-d}CI(\mathcal{R}_{d-1})$ in \eqref{eq:assumptionCI}.]{
  \includegraphics[width=.95\columnwidth]{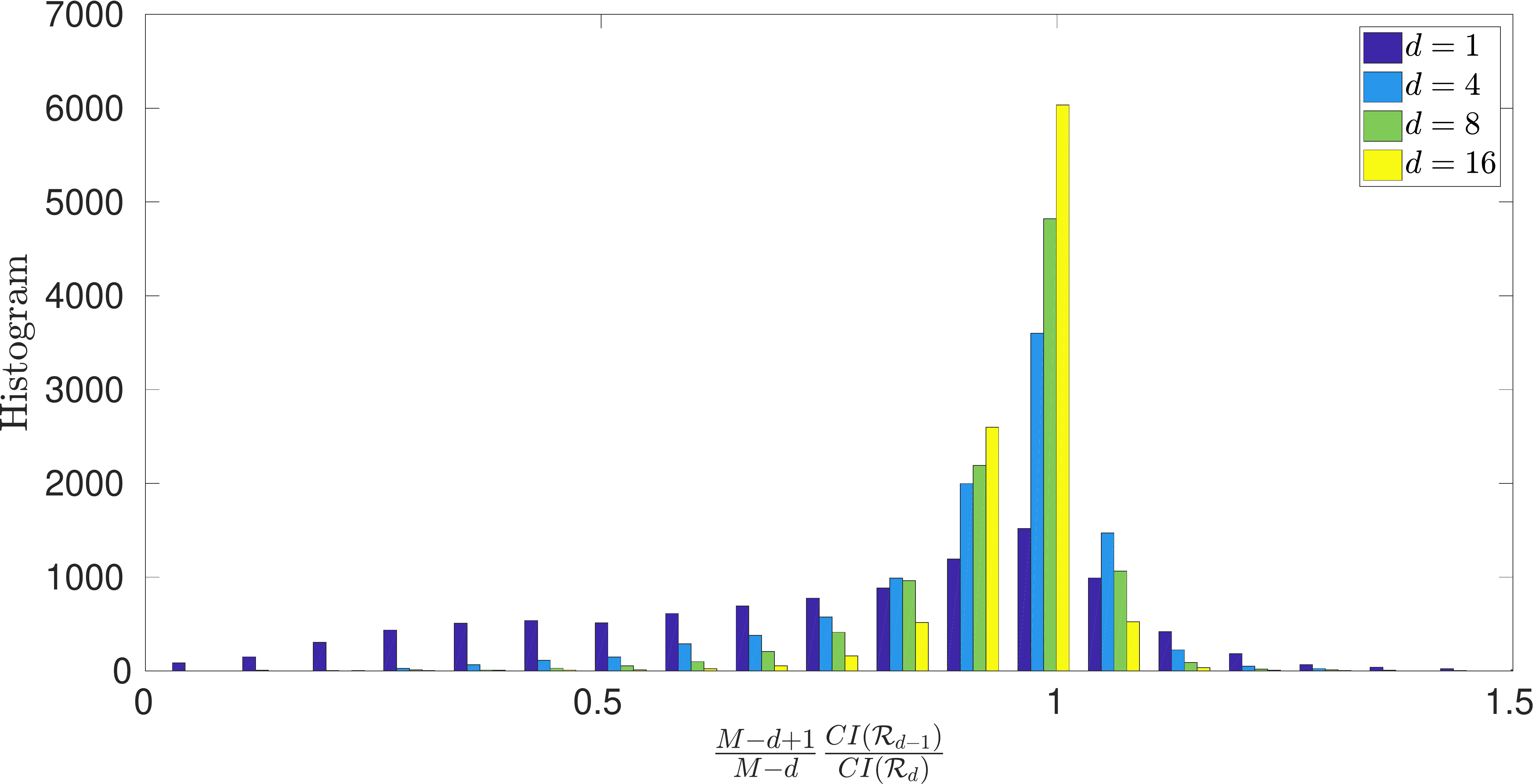}
  \label{fig:histogramCI}
 }
 \caption{Numerical validation of channel vector compressibility analysis.}
 \label{fig:CIfigs}
\end{figure}

In the derivation of \eqref{eq:rhoLB2} we relied on the assumption that for a set $\mathcal{R}_d$ containing $M-d$ elements, when $d\ll M-d$, if the largest element of the set is removed the CI increases by a factor $\frac{M-d+1}{M-d}$ with high probability. We depict the histogram of the normalized CI variation at different values of $d$ in Fig. \ref{fig:histogramCI} to show that this assumption is reasonable in the channels considered in our simulation  (the ratios are approximately equal to $1$ or less with high probability).

\subsection{Compressibility Index of random channels $\h_M$}

We defined the $\textnormal{CI}(\h_M)$ of an arbitrary vector $\h_M$ as our compressibility metric. We study the effect on the CI for random vector distributions where $\h_M$ is generated using four different models. We represent the compressibility for different channels in Fig. \ref{fig:CICDF}.

In all four cases we generate the delays as in the NYU mmWave MPC model \cite{Mathew2016}, using \eqref{eq:chanPb} with the same $\Pb_{\{\tau_\ell\}_{\ell=1}^{L}}$, and we adopt four different models for $\{\alpha_{\ell}\}_{\ell=1}^{L}$. First, we followed the NYU mmWave MPC model \cite{Mathew2016}, i.e. lognormal distributions where the mean depends on the delay. Second, we introduce a lognormal model without an exponential delay decay, i.e. $\log(\overline{\alpha}_{\ell})=\zeta_\ell$, where $\zeta_\ell$ is a normal with zero mean. Next, we replace the lognormal with an Rayleigh distribution to study the influence of the fourth moment of the distribution. In the third case $\overline{\alpha}_\ell$'s are generated as a Rayleigh distribution with adjusted delay-decay mean $e^{-\tau_\ell/\Gamma}$, and in the fourth case $\overline{\alpha}_\ell$ are generated as a Rayleigh with unit mean independent of $\tau_\ell$.

In Fig. \ref{fig:CICDF} we represent the c.d.f. of $\frac{M}{L}\mathrm{CI}(\h_M)$. We see that, as expected by the kurtosis of the Rayleigh distribution ($\kappa=2$), the adjusted CI satisfied $\frac{M}{L}\mathrm{CI}(\h_M)>.5$ for the 85\% percentile, in other words, comparable to a strictly sparse vector with $L/2$ equal-power coefficients. In comparison, both the introduction of delay-decaying means and the replacement of the Rayleigh distribution with a lognormal improved the ``compressibility'' of $\h_M$, reducing the 85\%-percentile adjusted CI to about $.2$. That is, comparable to a strictly sparse vector with $L/5$ equal-power coefficients. Finally, in the measurement-based  NYU mmWave MPC model, both the lognormal and exponential delay decay are adopted, and the 85\%-percentile of $\frac{M}{L}\mathrm{CI}(\h_M)$ is approximately $.1$, comparable to a strictly-sparse vector with $L/10$ equal non-zero coefficients.

Based on these CI observations, the MSE of any CS algorithm in the NYU mmWave MPC model will be much lower than  with Gaussian MPCs (Rayleigh amplitudes). This is verified in Fig. 4(b). For SNR$=0$ dB the MSE is 3 dB lower in the NYU mmWave MPC model than in a channel with Rayleigh distributed MPC amplitudes, while the intermediate models are $1.25$ dB and $2$ dB below the Rayleigh, respectively. This $3$ dB difference is numerically consistent with the quantitative prediction of our analysis, which shows that the MSE grows linearly with $\hat{L}$, and that, as we saw in Fig. 3(a), the value of $d$ such that $\Ex{}{\overline{\rho}(d)}=\frac{d}{N}\sigma^2$ with $\sigma^2=1$ is twice as high when $\alpha_\ell$ is Rayleigh distributed compared to the NYU mmWave MPC model.

\begin{figure}
 \centering 
 \subfigure[Adjusted CI C.D.F.]{
  \includegraphics[width=.95\columnwidth]{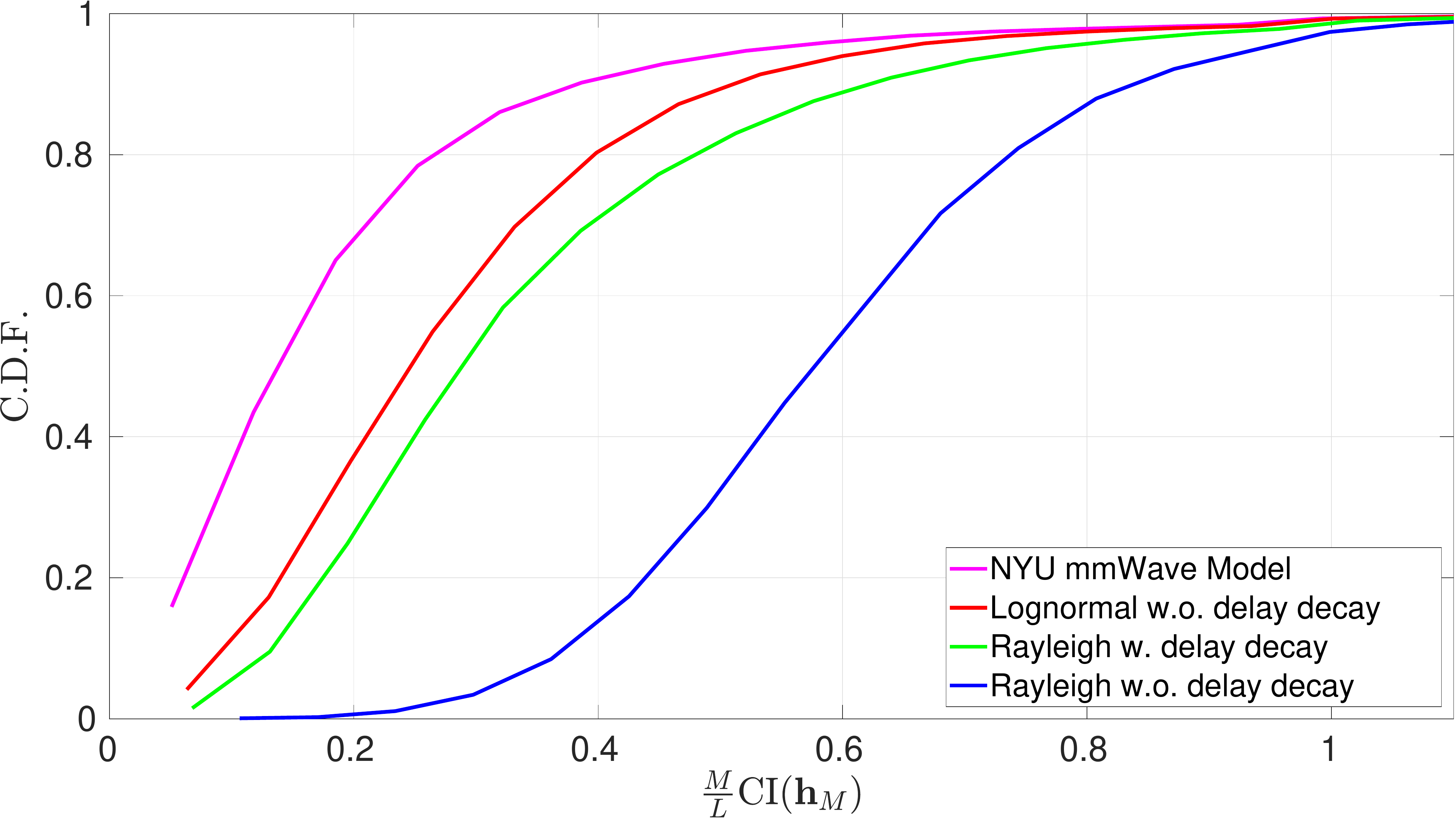}
 \label{fig:CICDF}
  }
 \subfigure[OMPBR MSE]{
  \includegraphics[width=.95\columnwidth]{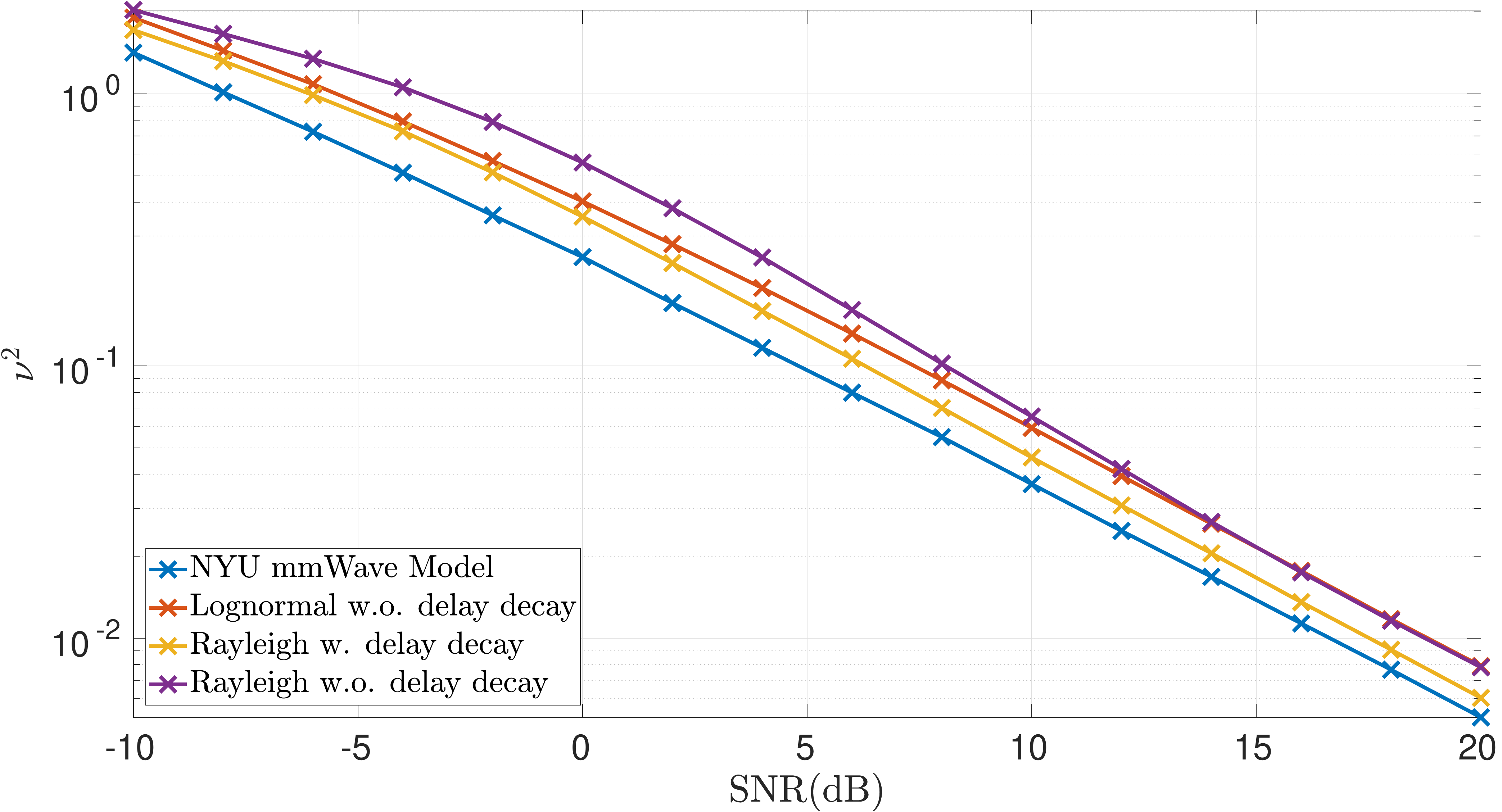}
 \label{fig:compressMSE}
  }
 \caption{Compressibility comparison for different channel models.}
 \label{fig:compress}
\end{figure}

\subsection{Greedy vs $\ell_1$ approaches}

In Fig. \ref{fig:L1fig} we represent the MSE for NYU mmWave vs Rayleigh $\{\alpha_\ell\}_{\ell=1}^L$, where we compare OMP, OMPBR and BPDN. BPDN solves $\min \|\hat{\bb}_{N_T}\|_1 \textnormal{ s.t. } \|\y_N-\Phib_{N_T}\hat{\bb}_{N_T}\|_2^2\leq \xi$ replacing the $\ell_0$ norm by the $\ell_1$ norm in \eqref{eq:bestsparse}. MATLAB's interior-point algorithm takes $10^4$ times more time than OMPBR in our simulation. BPDN can be adopted as ``direct'' solution or combined with LS, taking the indices of the ``direct'' solution above a threshold as a delay estimator $\hat{\Tau}_{\ell_1}$, and performing LS with matrix $\Pb_{\hat{\Tau}_{\ell_1}}$. OMPBR is only $0.86$ dB worse than BPDN for lognormal MPC amplitudes. This grows to $1.46$ dB for the Rayleigh case. The advantage of BPDN over OMP diminishes when $\overline\rho(d)$ decays faster. 
The impact on MSE of different channel models with the same algorithm was $3$ dB for both OMP and BPDN, which is greater than the MSE differences between OMP and BPDN for the same MPC amplitude model that we have observed.

\begin{figure}
 \centering 
 \subfigure[
 Lognormal
 ]{
  \includegraphics[width=.95\columnwidth]{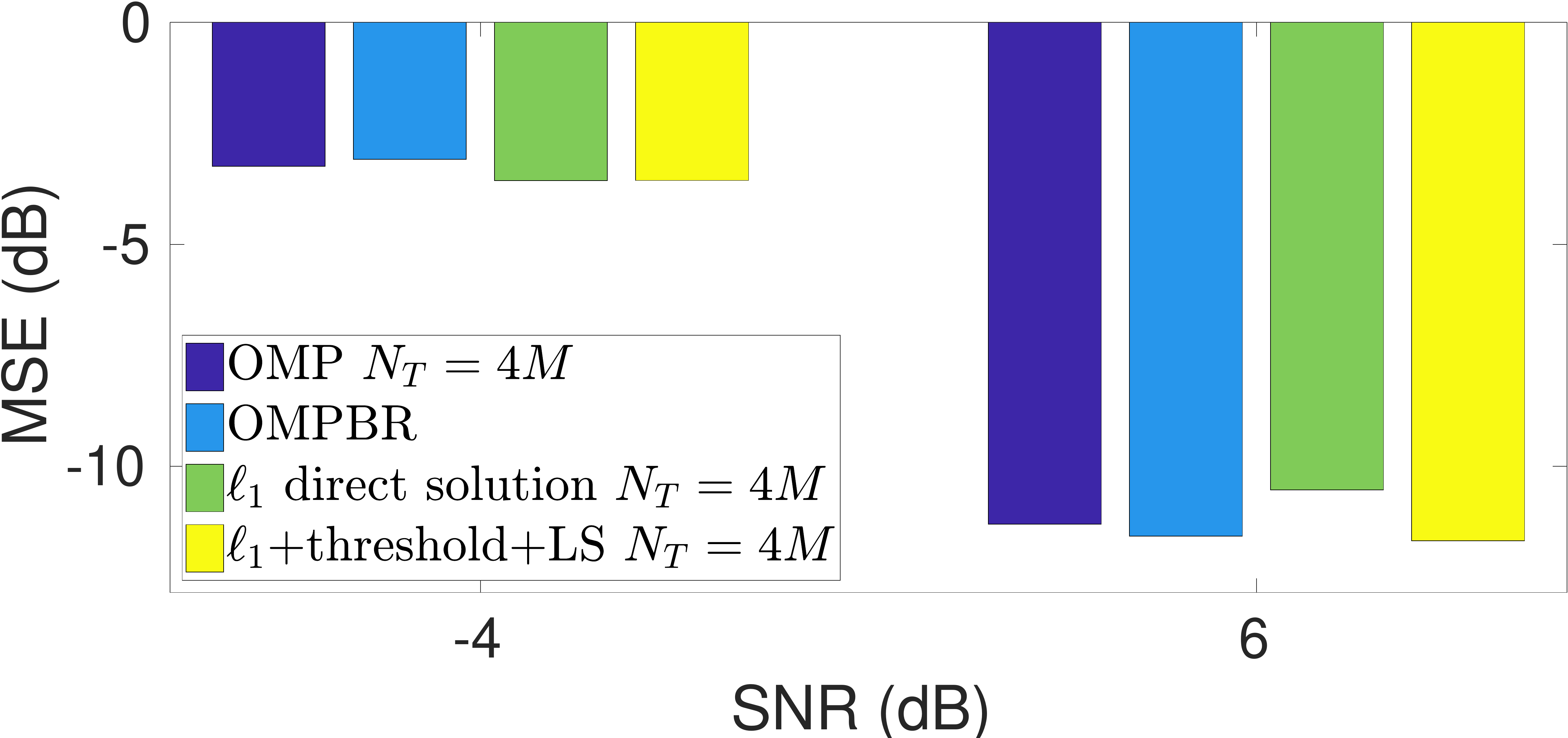}
  \label{fig:L1sfigL}
 }
 \hspace{.5in}
 \subfigure[Rayleigh]{
  \includegraphics[width=.95\columnwidth]{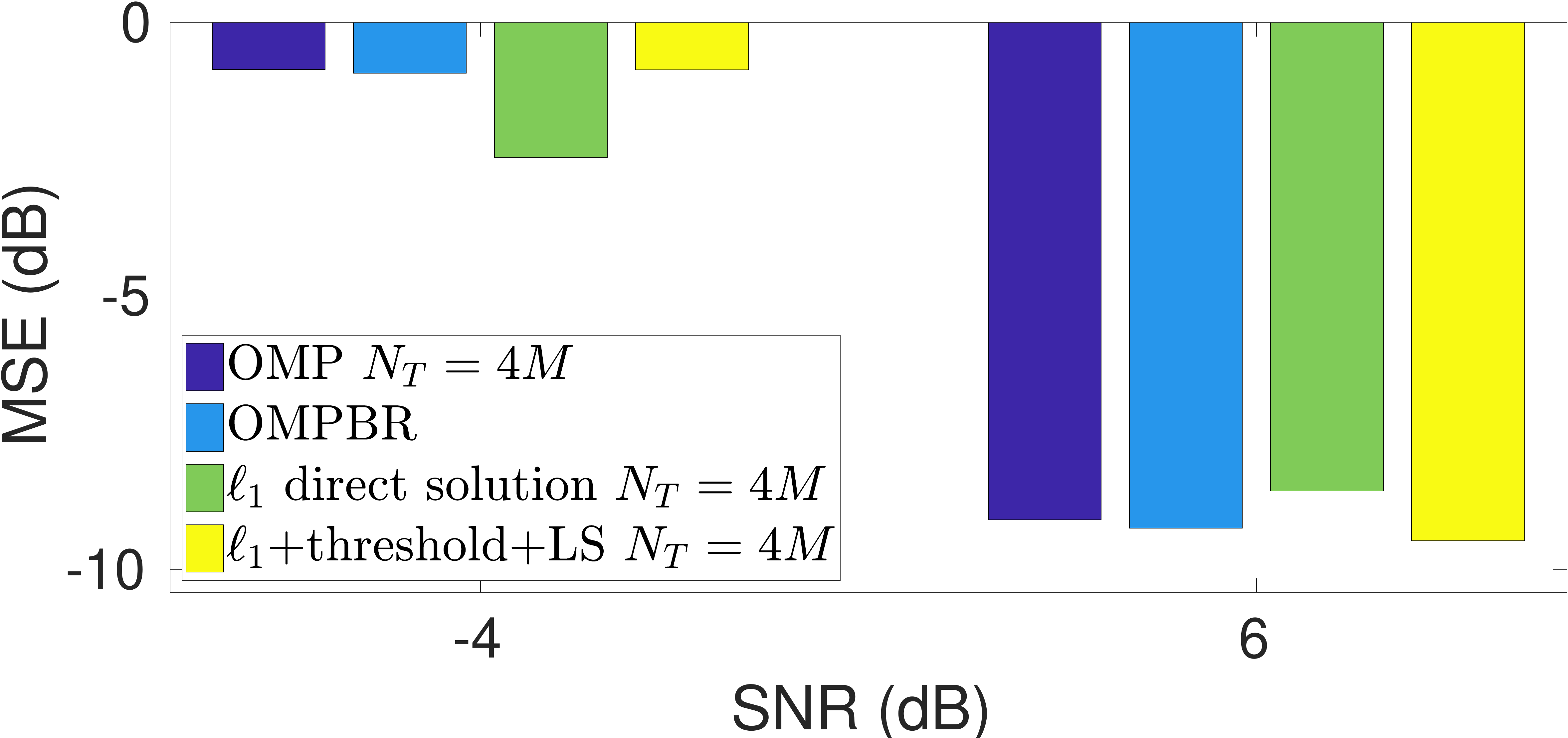}
  \label{fig:L1sfigR}
 }
 \caption{Numerical comparison of OMPBR and $\ell_1$ BPDN vs OMP with $N_T=4M$ in different channel models.}
 \label{fig:L1fig}
\end{figure}

\section{Effects on a MMSE OFDM Receiver}
\label{sec:receiver}

In this section we show how changes in the channel estimator affect the performance of an MMSE receiver. We assume the receiver uses an estimation of the channel modeled with an additive error $\hat{\h}_K=\h_K+\tilde{\h}_K$, where $\tilde{\h}_K$ is independent of the realizations of $\x$ and $\z$ in data subcarriers. The receiver knowns the statistics of the error $\Ex{\z,\h_K}{\tilde{\h}_K}=0$ (unbiased) and $\Ex{\z,\h_K}{\tilde{\h}_K\tilde{\h}_K^H}=\Sb_{\tilde{\h}_K}$. This holds both our LS benchmarks and for our hybrid estimator. The term ``perfect channel estimation'' refers to the special case $\tilde{\h}_K=0$ and $\hat{\h}_K=\h_K$.

We assume the receiver applies the MMSE linear equalizer to $\y$, defined as follows
\begin{equation}
\label{eq:mmsestatement}\B^*=\arg \min_{\B} \Ex{\x,\z,\tilde{\h}}{\|\B\y-\x\|^2}
\end{equation}
Denoting the statistics of $\x$ and $\z$ by $\Ex{\x}{\x}=0$, $\Ex{\x}{\x\x^H}=\Sb_\x$ and $\Ex{\z}{\z\z^H}=\sigma^2\I_K$ and using the assumption that $\x,\z,\tilde{\h}$ are independent, we derive this equalizer in Appendix \ref{app:mmse}:
\begin{equation}
\label{eq:mmsegen}
\B^*=(\Sb_{\x}\cdot\hat{\h}_K\hat{\h}_K^H+\Sb_{\x}\cdot\Sb_{\tilde{\h}_K}+\sigma^2\I_K)^{-1}\Sb_{\x}\D(\hat{\h}_K)^H.
\end{equation}

Examining \eqref{eq:mmsegen} we can deduce two qualitative influences of the channel estimator. First, the term $\Sb_{\x}\cdot\Sb_{\tilde{\h}_K}$ features the covariance of the error. We have shown the three estimation models present different covariance matrices  \eqref{eq:covMLMK},  \eqref{eq:covMLaK}, and \eqref{eq:covOMP}. This means that the ability of the MMSE equalizer \eqref{eq:mmsegen} to exploit correlations in the channel estimation error varies for the three estimation models. Second, $\x$ follows an i.i.d. distribution, we get $\Sb_{\x}=\I$ and $\Sb_{\x}\cdot\Sb_{\tilde{\h}_K}=\nu^2\I$. As a result, the non-diagonal terms in \eqref{eq:covMLMK},  \eqref{eq:covMLaK}, and \eqref{eq:covOMP} cease to affect \eqref{eq:mmsegen} when the inputs are i.i.d. If we wanted to exploit channel estimation error correlations, we would require non-independent inputs such as, for example, frequency domain coding.

We leave the design of such input coding strategies for future work. In this section, we present quantitative numerical performance results assuming $\x$ contains i.i.d. symbols and \eqref{eq:mmsegen} depends on $\nu^2$, but not on the non-diagonal coefficients of $\Sb_{\tilde{\h}}$.
For i.i.d. inputs $\B^*\y$ simplifies to a collection of independent scalar symbol-by-symbol channels at each subcarrier $k$ written as
\begin{equation}
\label{eq:mmsesimple}
b_{k}y_k=\frac{1}{|\hat{h}_k|^2+\nu^2+\sigma^2}\left[|\hat{h}_k|^2x_k-\hat{h}_k^H\tilde{h}_kx_k+\hat{h}_k^Hz_k\right]
\end{equation}
and we need only to design the detector of $x_k$ based on $b_{k}y_k$. Since the goal of this section is only to show that our channel estimation results have an effect on receiver performance, we leave the design of optimal detectors under imperfect channel estimation for future work. Instead, we demonstrate the performance when the receiver applies the optimal decision for the perfect channel estimation case, which are the constellation demodulators built in MATLAB.

The Bit Error Rate (BER)  vs SNR for QPSK and 16QAM inputs in the OFDM system with the equalizer \eqref{eq:mmsesimple} is represented in Fig \ref{fig:ber128}. We represent the Rayleigh channel BER for reference as a dashed black line. Results for this channel show the BER decays as $\textnormal{SNR}^{-1}$ \cite{goldsmith2005book}. We represent the BER on a  NYU mmWave MPC model with perfect channel state information as a solid black line, noting that the BER curve displays a similar shape about 1dB lower.

The magenta curve represents the BER with the non-sparse ML estimator benchmark, which is $3$ dB worse than the perfect CSI. This is consistent with the fact that for $N=M$ we have $\nu^2_{LS-M}=\sigma^2$ in \eqref{eq:errMLMK}. We adopted a decisor that treats estimation error as noise, and hence it behaves as if the ``effective noise power'' was $\nu^2_{LS-M}+\sigma^2=2\sigma^2$. We represent in red the BER with the genie-aided ML sparse estimator benchmark. A gain of $2.1$ dB versus the non-sparse ML estimator \eqref{eq:MLM}, \eqref{eq:MLMK} and a gap of $.9$ dB versus perfect channel knowledge is achieved, which is consistent with  \eqref{eq:errMLaK} and an ``effective noise power'' $(1+L/N)\sigma^2$ with $L/N=28/128$.

We consider again OMP with $N_T=M$, $N_T=4M$ and OMPBR. At low SNR the three OMP variants achieve an almost identical BER than the genie-aided sparse benchmark, and a gain of $2$ dB with regard to the non-sparse estimator benchmark. This is because at low SNR $\hat{L}<L/2$ and the error with OMP is comparable to the genie aided benchmark. At high SNR the gain of OMPBR drops down to $1.4$ dB. This is consistent with the interpretation of $\nu^2+\sigma^2$ as ``effective noise power'',  which for $\hat{L}=L=28$ predicts a gap between OMPBR and the sparse genie-aided benchmark of at least $10\log_{10}\left(\frac{1+2\frac{28}{128}}{1+\frac{28}{128}}\right)=.71$ dB. Moreover, OMP with $N_T=M$ performs even worse due to the limitations of its  dictionary.
 
\begin{figure}
 \centering
 \subfigure[QPSK]{
  \includegraphics[width=.95\columnwidth]{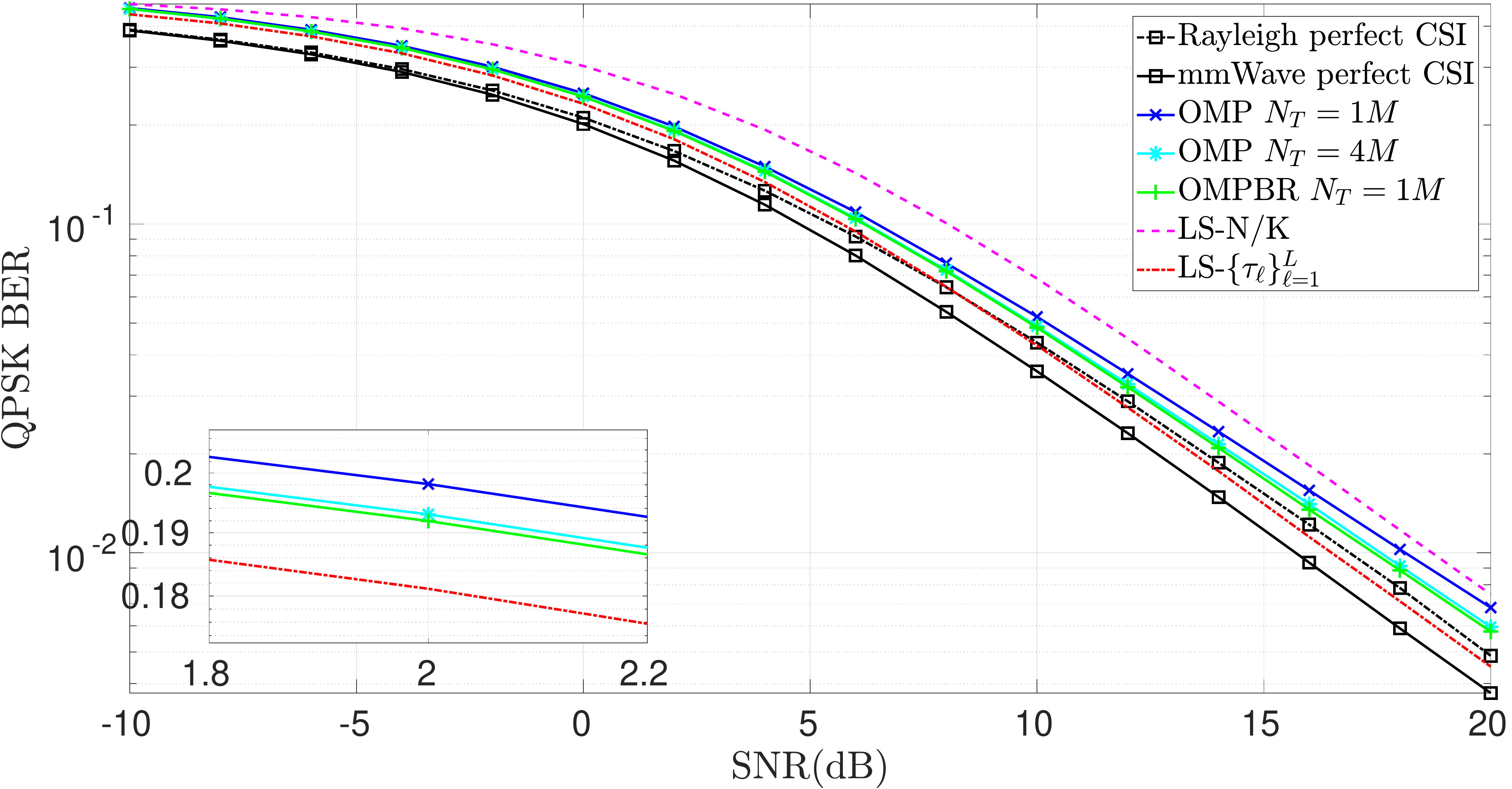}
 \label{fig:ber128qpsk}
  }
 \subfigure[16-QAM]{
  \includegraphics[width=.95\columnwidth]{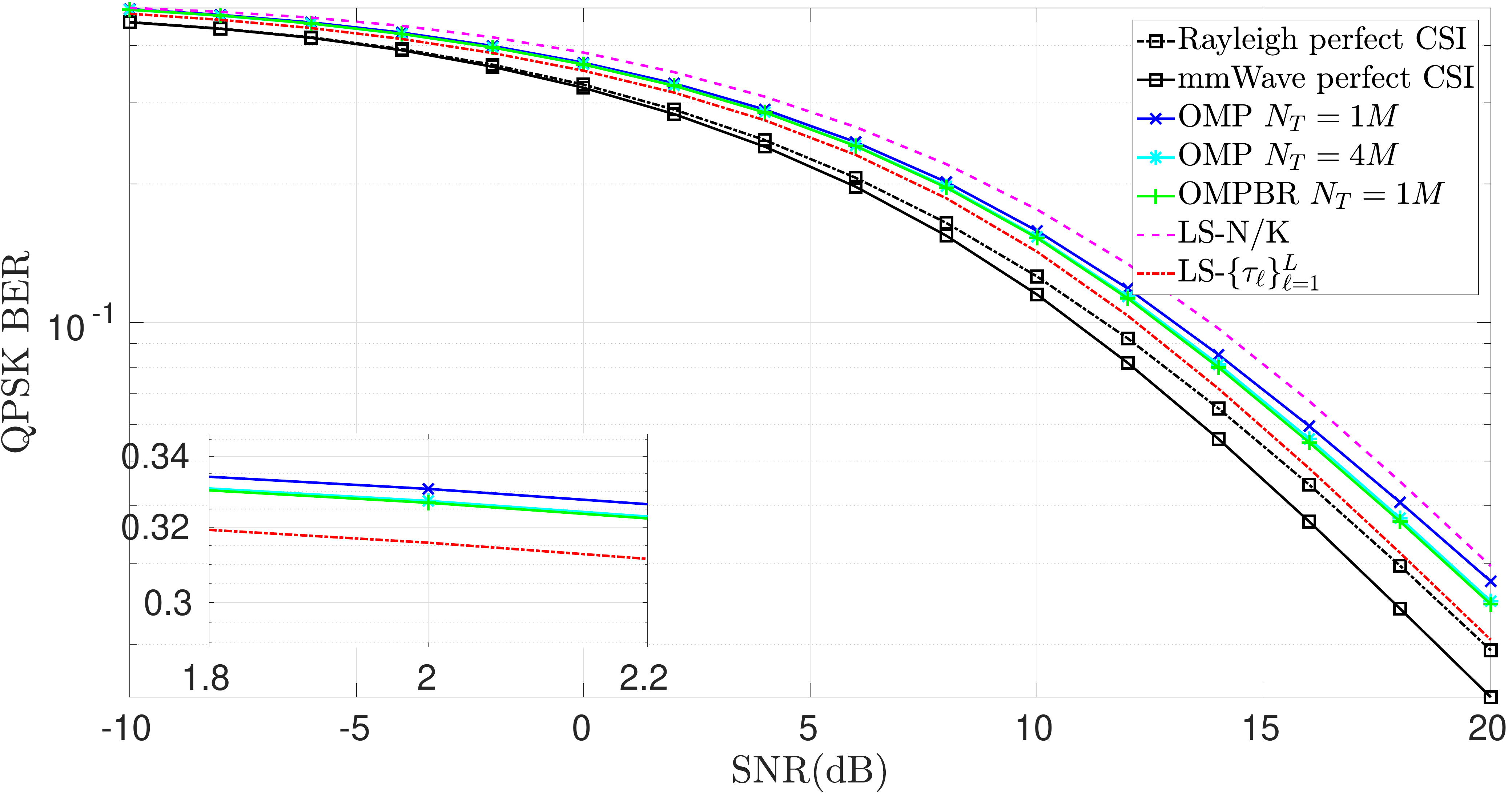}
 \label{fig:ber12816qam}
  }
  \caption{BER with $N_{p}=128$ pilots.}
 \label{fig:ber128}
\end{figure}

\section{Conclusions and Future Work}
\label{sec:conclusions}

We have analyzed the impact of non-Gaussian MPC amplitude distributions in CS sparse channel estimation. We have shown that in sparse channels where the MPC amplitude distribution has a high fourth moment, the channel power concentrates in a subset of \textit{dominant MPCs}, which are the only MPCs that any CS algorithm must retrieve. We have also shown that the estimation error of OMP is proportional to the number of retrieved MPCs (iterations). Therefore in non-Gaussian MPC amplitude models with high fourth moment OMP achieves a lower error and performs much better, equaling the performance of BPDN. Our compressibility analysis extends an existing study of the compressibility of random i.i.d. vectors. We used the Compressibility Index to score the compressibility of arbitrary channel vectors. Our analysis characterizes an ``oracle'' lower bound that holds for any CS algorithm and is not particular to OMP. The analysis shows that the higher the fourth moment of the MPC amplitude distribution, the lower the CI of the channel, the fewer dominant MPCs there are, and the lower the MSE of CS estimation is. 

Among others, this affects mmWave channel models with lognormal MPC amplitudes. These models present ``statistical sparsity'' in the sense that the set of MPCs has low CI. This makes mmWave channel CS estimation MSE lower than in more popular Gaussian MPC models. Improvement of channel estimation MSE directly translates into the performance of a receiver. For ease of exposition, we presented our analysis for a single antenna system model where only delay-domain sparse channel estimation is considered. However, our results can be extended to hybrid analog-digital MIMO OFDM mmWave channel estimation as in \cite{7953407,rodriguez2017frequency,Mo2017,Venugopal2017}.
 
\appendices

\begin{figure*}[!t]
\normalsize
\setcounter{MYtempeqncnt}{\value{equation}}
\setcounter{equation}{41}

\begin{equation}
\label{eq:resexpanded}
 \begin{split}
    \|\rr_{\hat{L}}\|^2&=N\rho(\hat{L})+\|\z_N\|^2-\frac{N}{K}\|\tilde{\h}_S\|^2\\
    &\quad+2\Re\{\h_M^H[(\I_M-\Pb_{\hat{\Tau}}\Pb_{\hat{\Tau}}^\dag)^H\F_{N/K,M}^H(\I-\F_{N/K,M}\Pb_{\hat{\Tau}}(\F_{N/K,M}\Pb_{\hat{\Tau}})^\dag)]\z_N\}\\
 \end{split}
\end{equation}
\setcounter{equation}{\value{MYtempeqncnt}}

\hrulefill
\vspace*{4pt}
\end{figure*}

\section{}
\label{app:firsterm}

We define $\breve{\z}=\sqrt{\frac{K}{N}}\F_{N/K,M}^H\z_N$, where if $\z_N$ is AWGN, then $\breve{\z}$ is also AWGN with covariance $\sigma^2\I_M$. Since $\Pb_{\hat{\Tau}}(\hat{\bb}-\bb)=\Pb_{\hat{\Tau}}(\F_{N/K}\Pb_{\hat{\Tau}})^\dag\y_{N}-\Pb_{\hat{\Tau}}(\F_{N/K}\Pb_{\hat{\Tau}})^\dag\h_N=\Pb_{\hat{\Tau}}\Pb_{\hat{\Tau}}^{\dag}\sqrt{\frac{K}{N}}\breve{\z}$, we can write
\begin{equation}
\begin{split}\frac{\Ex{\z}{\|\tilde{\h}_S\|^2}}{K}&=\frac{\Ex{\z}{\|\Pb_{\hat{\Tau}}(\hat{\bb}-\bb)\|^2}}{K}\\
&=\frac{\Ex{\breve{\z}}{\breve{\z}^H\Pb_{\hat{\Tau}}(\Pb_{\hat{\Tau}}^H\Pb_{\hat{\Tau}})^{-1}\Pb_{\hat{\Tau}}^H\breve{\z}}}{N}\\
&=\frac{\Ex{\breve{\z}}{\tr\{\Pb_{\hat{\Tau}}^H\breve{\z}\breve{\z}^H\Pb_{\hat{\Tau}}(\Pb_{\hat{\Tau}}^H\Pb_{\hat{\Tau}})^{-1}\}}}{N}.\\\end{split}
\end{equation}

Here $\Pb_{\hat{\Tau}}$ depends on $\breve{\z}$. If $\Pb_{\hat{\Tau}}$ was independent of $\breve{\z}$, we would be able to write a closed form expression as we had for $\Pb_{\{\tau_\ell\}_{\ell=1}^{L}}$ in \eqref{eq:errMLaK}. However, in OMP $\Pb_{\hat{\Tau}}$ depends on $\breve{\z}$ via $\y_{N/K}$. Still, we can note that  in the limit as $\sigma^2$ approaches zero the dependency vanishes, and thus
\begin{equation}
\begin{split}
\frac{\Ex{\z}{\|\tilde{\h}_S\|^2}}{K}&\stackrel{\sigma^2\ll1}{\xrightarrow{\hspace{.7cm}}}\frac{\tr\{\Pb_{\hat{\Tau}}^H\Ex{\breve{\z}}{\breve{\z}\breve{\z}^H}\Pb_{\hat{\Tau}}(\Pb_{\hat{\Tau}}^H\Pb_{\hat{\Tau}})^{-1}\}}{N}\\
&=\frac{\Ex{\z}{\hat{L}}\sigma^2}{N}.\\
\end{split}
\end{equation}

We convert this asymptotic result into a lower bound for the case of OMP with an orthogonal dictionary, assuming that $N_T=M$ and that $p(t)$ is a Nyquist pulse. This leads to $\Pb_{N_T}=\I_M$ and the matrix $\Pb_{\hat{\Tau}}$ is a subset of the columns of the identity.  We denote a diagonal ``selection mask'' matrix $\D(\hat{\Tau})=\Pb_{\hat{\Tau}}\Pb_{\hat{\Tau}}^H$ which satisfies $D_{n,n}=1$ if $nT\in\hat{\Tau}$. Using this matrix and $(\F_{N/K,M}\Pb_{\hat{\Tau}})^{\dag}=\Pb_{\hat{\Tau}}^{\dag}\frac{K}{N}\F_{N/K,M}^H$, the OMP estimator takes the form of a diagonal selection mask multiplied by a pre-calculated vector. Conveniently, this pre-calculated term is identical to \eqref{eq:MLM}.
\begin{equation}
\label{eq:orthLS}
\hat{\bb}_{N_T}=\Pb_{\hat{\Tau}}\frac{K}{N}\Pb_{\hat{\Tau}}^{\dag}\F_{N/K,M}^H\y_N=\D(\hat{\Tau})\hat{\h}_{M}^{ML-M}.
\end{equation}

Using this notation we can write the first error term as a function of $\D(\hat{\Tau})$ as $\frac{\|\tilde{\h}_S\|^2}{K}=\frac{\|\D(\hat{\Tau})\breve{\z}\|^2}{N}$. Examining \eqref{eq:orthLS} and Alg. \ref{alg:OMP} we note that the $n$-th step of the greedy OMP algorithm adds to $\hat{\Tau}$ the $n$-th largest elements of the vector $\hat{\h}_{M}^{ML-M}$. Thus the probability of the event $D_{n,n}=1$ corresponds to the probability that $\hat{h}_{M}^{ML-M}(n)$ is one of the $\hat{L}$ largest elements of the vector $\hat{\h}_{M}^{ML-M}$. Since $\breve{\z}$ is AWGN, the variable $\hat{\h}_{M}^{ML-M}=\h_M+\breve{\z}$ is a Gaussian centered at $\h_M$. From this we deduce that the covariance of $D_{n,n}$ and $|\breve{z}_n|^2$ is non-negative. Therefore

\begin{equation}
\label{eq:1sterrterm}
 \begin{split}
  \frac{\Ex{\z}{\|\D(\hat{\Tau}_{\hat{L}})\breve{\z}\|^2}}{N}&=\frac{\sum_{n=1}^{M}\Ex{\z}{D_{n,n}|\breve{z}_n|^2}}{N}\\
&\stackrel{a}{\geq}
  \frac{\sum_{n=1}^{M}\Ex{\z}{D_{n,n}}\Ex{\z}{|\breve{z}_n|^2}}{N}\\
&=\frac{\Ex{\z}{\hat{L}}}{N}\sigma^2
 \end{split}
\end{equation}
where $(a)$ follows from the non-negative covariance and would be an equality if $D_{n,n}$ was independent of $\breve{z}_n$. This shows that indeed $\frac{\Ex{\z}{\hat{L}}}{N}\sigma^2$ is a lower bound on $\Ex{\z}{\frac{\|\tilde{\h}_S\|^2}{K}}$ for $N_T=M$ and $p(t)$ a Nyquist pulse. This lower bound is tight for $\sigma^2\to0$ as $D_{n,n}$ is independent of $\breve{z}_n$.

\section{}
\label{app:secondterm}

OMP stops at the first iteration $i$ that satisfies the stop condition $\|\rr_{i}\|^2\leq \xi$ where the residual $\rr_{i}$ is defined by lines 5 and 12 of Alg \ref{alg:OMP}. 
Using $\frac{K}{N}\F_{N/K,M}^H\F_{N/K,M}=\I_M$ and $\y_N=\F_{N/K,M}\h_M+\z_N$, the residual $\rr_{\hat{L}}\triangleq (\I_N-\F_{N/K,M}\Pb_{\hat{\Tau}}\Pb_{\hat{\Tau}}^\dag\frac{K}{N}\F_{N/K,M}^H)\y_N$ can be written as
\begin{equation}
 \begin{split}
    \rr_{\hat{L}}&=\F_{N/K,M}(\I_M-\Pb_{\hat{\Tau}}\Pb_{\hat{\Tau}}^\dag)\h_M\\
    &\quad+(\I -\F_{N/K,M}\Pb_{\hat{\Tau}}(\F_{N/K,M}\Pb_{\hat{\Tau}})^\dag)\z_N.\\    
 \end{split}
\end{equation}
Taking the magnitude, we apply the definition of $\rho(\hat{L})$. Since $\F_{N/K,M}^H\F_{N/K,M}=\frac{N}{K}\I_M$, $\|\F_{N/K,M}\Pb_{\hat{\Tau}}(\F_{N/K,M}\Pb_{\hat{\Tau}})^\dag\z_N\|^2=\frac{N}{K}\|\tilde{\h}_S\|^2$ and since there is an orthogonal projection we have $\|\z_N\|^2=\|(\I -\F_{N/K,M}\Pb_{\hat{\Tau}}(\F_{N/K,M}\Pb_{\hat{\Tau}})^\dag)\z_N\|^2+\frac{N}{K}\|\tilde{\h}_S\|^2$, so we can write \eqref{eq:resexpanded} where we can ignore the last term because the average of $\z_N$ is zero.

\stepcounter{equation}

We observe that in each $i$-th iteration OMP reduces the residual $\|\rr_{i}\|^2$ so that $\hat{L}$ is the index of the first iteration where the stop condition is met. Thus, if we make the assumption that $\|\rr_{i}\|^2$ decreases smoothly in the sense that $\|\rr_i\|^2-\|\rr_{i-1}\|^2$ is small, this means that we can approximate the stop condition as an equality at $i=\hat{L}$, and we get $\|\rr_{\hat{L}}\|^2\simeq \xi$. Choosing $\xi=N\sigma^2$ we have
\begin{equation}
 \begin{split} 
  \Ex{\z}{\rho(\hat{L})}&\simeq\frac{N\sigma^2-\Ex{\z}{\|\z_N\|^2}+\frac{N}{K}\Ex{\z}{\|\tilde{\h}_S\|^2}}{N}\\
  &\quad-\frac{\Ex{\z}{2\Re\{\h_M^H[\dots]\z_N\}}}{N}\\
  &=\Ex{\z}{\|\tilde{\h}_S\|^2}/K
 \end{split}
\end{equation}

\begin{figure*}[!t]
\normalsize
\setcounter{MYtempeqncnt}{\value{equation}}
\setcounter{equation}{44}

\begin{equation}
\label{eq:mmseasprod}
\B_{MMSE} =\min_{\B} \Ex{\x,\tilde{\h},\z}{(\B\D(\hat{\h})\x+\B\D(\tilde{\h})\x+\B\z-\x)^H(\B\D(\hat{\h})\x+\B\D(\tilde{\h})\x+\B\z-\x)}
\end{equation}

\setcounter{equation}{46}

\begin{equation}
\label{eq:16termmmse3}
\min_{\B} \tr\Big\{(\Sb_{\x}\cdot\hat{\h}\hat{\h}^H)\B^H\B-\Sb_{\x}\D(\hat{\h})^H\B^H+(\Sb_{\x}\cdot\Sb_{\tilde{\h}_K})\B^H\B+\sigma^2\B^H\B-\B\D(\hat{\h})\Sb_{\x}+\Sb_{\x}\Big\}.
\end{equation}
\setcounter{equation}{\value{MYtempeqncnt}}

\hrulefill
\vspace*{4pt}
\end{figure*}

\section{}
\label{app:otherworks}

In this appendix we discuss other works that have also reported connections between the fourth moment and compressibility or sequence unevenness. This motivates using our fourth-moment based CI metric as our measure of channel compressibility.

In \cite{Gribonval2012} an exact calculation of $\overline{\rho}(d)$ is provided for large i.i.d. vectors. In essence, Appendix A of \cite{Gribonval2012} shows that if $\vv$ is a size-$M$ i.i.d. vector where each coefficient of magnitude $|v_m|^2$ has a CDF $F_{|v|^2}(y)$, by the CLT we have 
\begin{equation}
\label{eq:intFv}
\lim_{M\to\infty} \overline{\rho}(d) = 1-\frac{\int_{0}^{\delta} yF_{|v|^2}(y)dy}{\Ex{}{|v|^2}}
\end{equation}
where $\delta=F_{|v|^2}^{-1}(1-\frac{d}{M})$. This expression depends on $F_{|v|^2}(y)$, and $\overline{\rho}(d)$ will decay fast if the variance of $F_{|v|^2}(y)$ is large. If $\Ex{}{v}=0$ this is related to the fourth moment of $v$ by $\frac{\textnormal{Var}[|v|^2]}{\Ex{}{|v|^2}^2} = \kappa(v)-1$.

In narrowband wide-sense stationary uncorrelated scattering time-varying fading channels, the Amount of Fading (AF) is a measure of channel variability \cite{Rao2015}. If we denote the channel gain by $h(t)$, with $\Ex{}{h(t)}=0$ and an i.i.d. distribution $f(h(t))$ for each $t$, the AF is $AF=\frac{\textnormal{Var}[|h(t)|^2]}{\Ex{}{|h(t)|^2}^2}$. Thus $AF=\kappa(h(t))-1$, and again we observe a relation to the fourth moment.

Finally, \cite{fgomezUnified} studied non-coherent wideband fading channel capacity. The expression \cite[Eq. (9)]{fgomezUnified} is a ``maximum bandwidth threshold'' above which capacity degrades due to channel uncertainty. This threshold grows with the channel kurtosis, suggesting that heavy-tailed channel distributions are associated with less capacity degradation due to channel uncertainty.

\section{}
\label{app:mmse}
We use $\D(\x)\h=\D(\h)\x$ and start by replacing $\y$ into \eqref{eq:mmsestatement} and expanding the product, producing \eqref{eq:mmseasprod}.

\stepcounter{equation}

The first product in \eqref{eq:mmseasprod} is $\Ex{\x,\tilde{\h},\z}{\x^H\D(\hat{\h})^H\B^H\B\D(\hat{\h})\x}$. Using the trace rotation we get
\begin{equation}
 \begin{split}
  &\Ex{\x,\tilde{\h},\z}{\tr\{\D(\hat{\h})\x\x^H\D(\hat{\h})^H\B^H\B\}}\\
  &\quad\stackrel{a}{=}\tr\{\Ex{\x,\tilde{\h},\z}{(\x\x^H\cdot\hat{\h}\hat{\h}^H)\B^H\B}\}\\
  &\quad\stackrel{b}{=}\tr\{(\Sb_{\x}\cdot\hat{\h}\hat{\h}^H)\B^H\B\},
 \end{split}
\end{equation}
where $(a)$ stems from $\D(\hat\h)\x=\hat\h\cdot\x$, and $(\x\cdot\hat\h)(\x\cdot\hat\h)^H=\x\x^H\cdot\hat\h\hat\h^H$ and $(b)$ from the independency of $\x$, $\z$, and $\tilde{\h}$. Assuming $\tilde{\h}$, $\z$ and $\x$ are zero-mean, and $\Ex{}{\x\x^H}=\Sb_\x$, $\Ex{}{\z\z^H}=\sigma^2\I$, and $\Ex{}{\tilde{\h}\tilde{\h}^H}=\Sb_{\tilde{\h}}$, we apply the same steps to all 16 cross-products in \eqref{eq:mmseasprod} producing \eqref{eq:16termmmse3}.

\stepcounter{equation}

Defining $M_1=(\Sb_{\x}\cdot\hat{\h}\hat{\h}^H)+(\Sb_{\x}\cdot\Sb_{\tilde{\h}})+\sigma^2\I$ and $M_2=\D(\hat{\h})\Sb_{\x},$ we must solve
$\displaystyle\min_{\B}\tr\left\{ M_1\B^H\B-\B M_2-(\B M_2)^H+\Sb_{\x}\right\}.$ This problem is well known in the perfect-CSI case, with solution $\B_{MMSE}=M_1^{-1}M_2^H.$ Substituting $M_1$ and $M_2$ concludes the proof of \eqref{eq:mmsegen}.



\begin{IEEEbiography}
[{\includegraphics[width=1in,height=1.25in,clip,keepaspectratio]{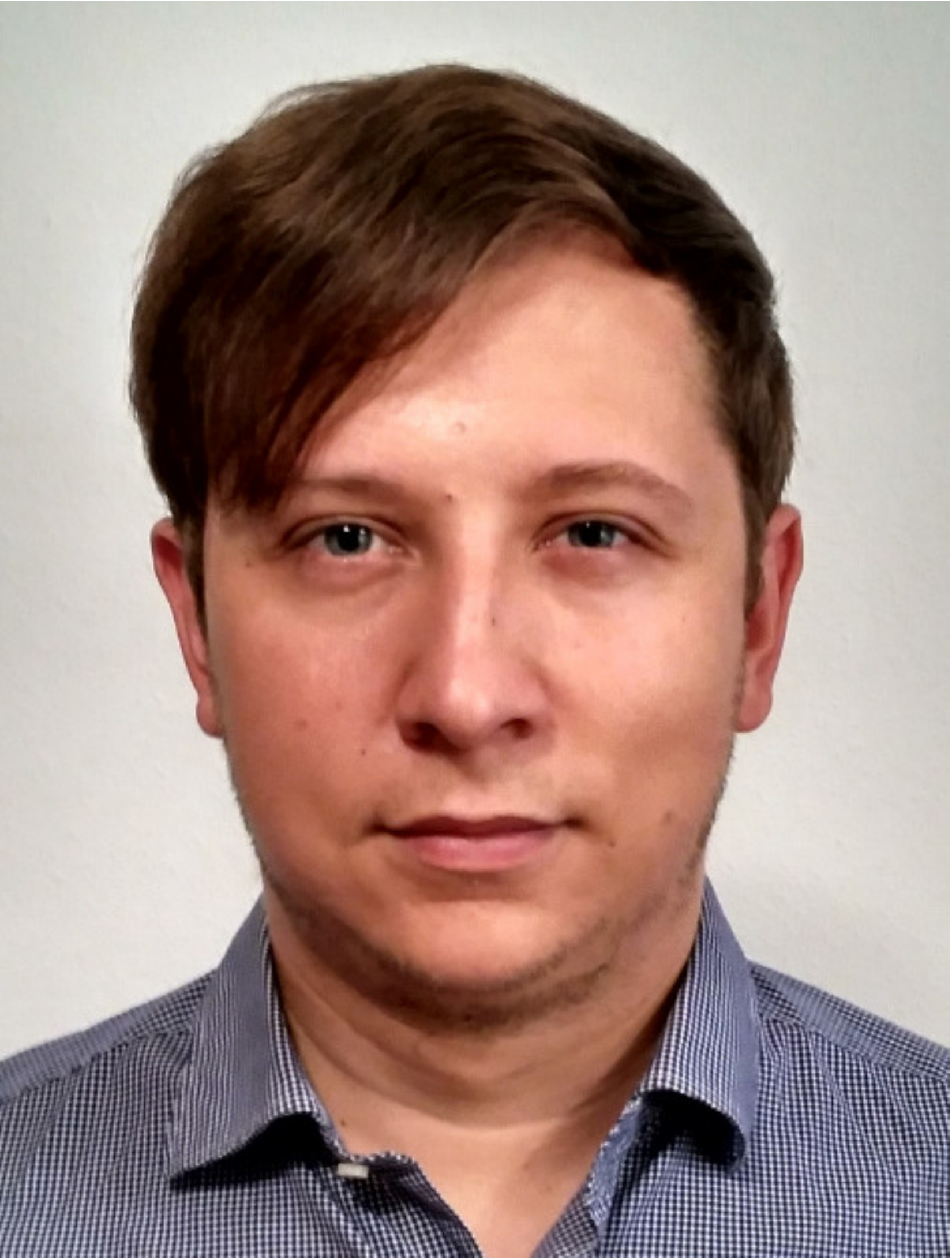}}]
{Felipe G\'omez-Cuba}
received the Ingeniero de Telecomunicaci\'on degree in 2010, M.Sc in Signal Processing Applications for Communications in 2012, and a PhD degree in 2015 from the University of Vigo, Spain. He worked as a researcher in the Information Technologies Group (GTI), University of Vigo, (2010--2011), the Galician Research and Development center In Advanced Telecommunications (GRADIANT),  (2011--2013), the NYUWireless center at NYU Tandon School of Engineering (2013--2014) and in University of Vigo with the FPU grant from the Spanish MINECO (2013--2016). He has a Marie Curie Individual Fellowship - Global Fellowship with the Dipartimento d'Engegneria dell'Informazione, University of Padova, Italy (2016-present) and he was a postdoctoral scholar with the Department of Electrical Engineering, Stanford University, USA (2016-2018). His main research interests are new paradigms in wireless networks such as cooperative communications, ultra high density cellular systems, wideband communications and massive MIMO.
\end{IEEEbiography}

\begin{IEEEbiography}
[{\includegraphics[width=1in,height=1.25in,clip,keepaspectratio]{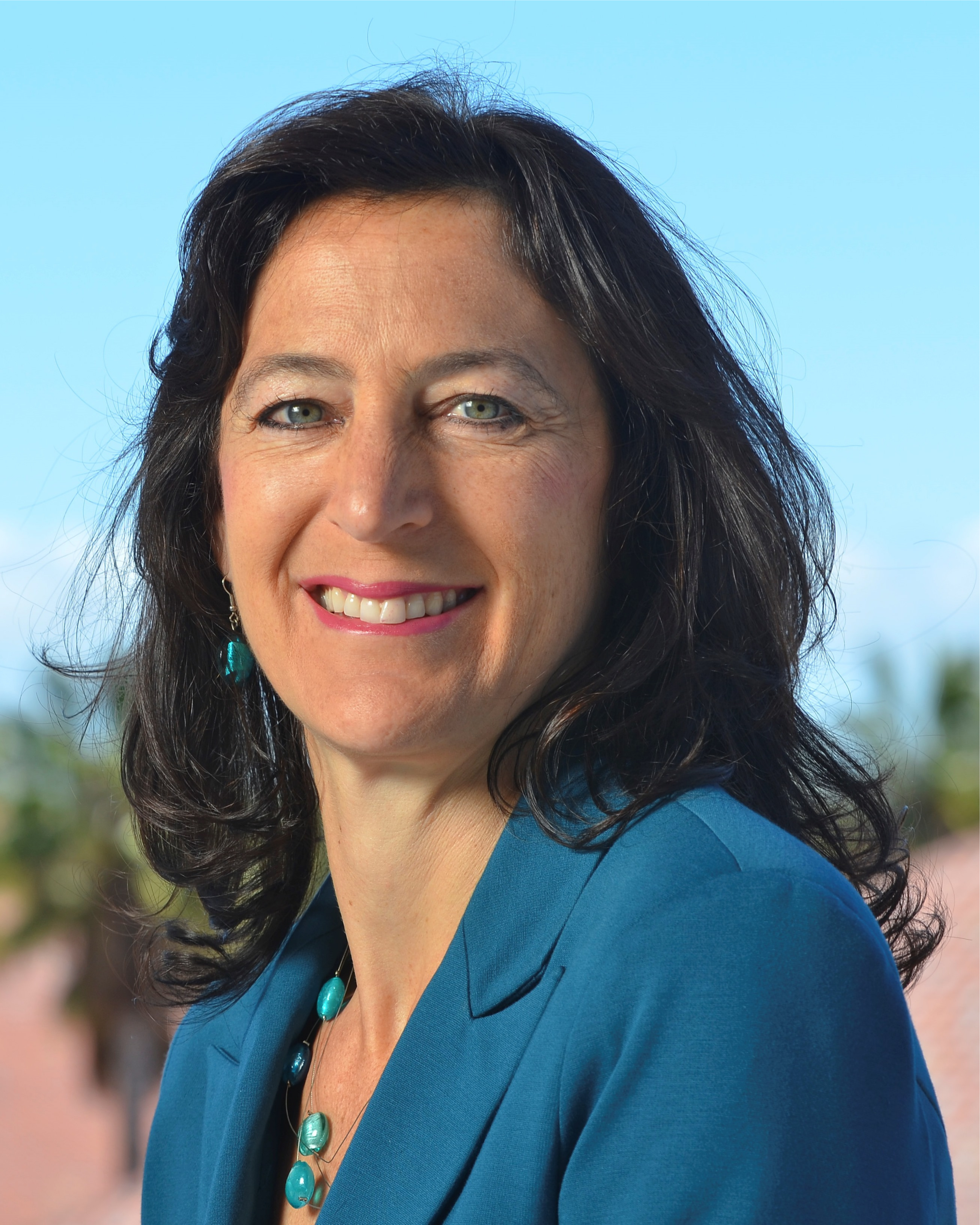}}]
{Andrea Goldsmith}
is the Stephen Harris professor in the School of Engineering and a professor of
Electrical Engineering at Stanford University. Her research interests are in information theory,
communication theory, and signal processing, and their application to wireless communications,
interconnected systems, and neuroscience. She founded and served as Chief Technical Officer of Plume
WiFi (formerly Accelera, Inc.) and of Quantenna (QTNA), Inc, and she currently serves on the Board of
Directors for Medtronic (MDT) and Crown Castle Inc (CCI). Dr. Goldsmith is a member of the National
Academy of Engineering and the American Academy of Arts and Sciences, a Fellow of the IEEE and of
Stanford, and has received several awards for her work, including the IEEE Sumner Technical Field
Award, the ACM Athena Lecturer Award, the ComSoc Armstrong Technical Achievement Award, the
WICE Mentoring Award, and the Silicon Valley/San Jose Business Journal’s Women of Influence Award.
She is author of the book ``Wireless Communications and co-author of the books ``MIMO Wireless
Communications and “Principles of Cognitive Radio,” all published by Cambridge University Press, as
well as an inventor on 29 patents. She received the B.S., M.S. and Ph.D. degrees in Electrical Engineering
from U.C. Berkeley.

Dr. Goldsmith is currently the founding Chair of the IEEE Board of Directors Committee on Diversity,
Inclusion, and Ethics. She served as President of the IEEE Information Theory Society in 2009 and as
founding Chair of its student committee. She has also served on the Board of Governors for both the
IEEE Information Theory and Communications Societies. At Stanford she has served as Chair of
Stanford’s Faculty Senate and for multiple terms as a Senator, and on its Academic Council Advisory
Board, Budget Group, Committee on Research, Planning and Policy Board, Commissions on Graduate
and on Undergraduate Education, Faculty Women’s Forum Steering Committee, and Task Force on
Women and Leadership.
\end{IEEEbiography}

\end{document}